\documentclass[11pt]{article}
\usepackage{amssymb,amsthm,amsmath}

\title{%
  Hop-Spanners for Geometric Intersection Graphs
  \thanks{Research supported in part by NSF award DMS-0701280 and the Summer Scholars Program at Tufts University.}
}

\author{%
Jonathan B. Conroy%
  \thanks{Department of Computer Science, Dartmouth College, Hanover, NH, USA. 
          Email: \texttt{Jonathan.Conroy.GR@dartmouth.edu}}
\and
Csaba D. T\'oth%
  \thanks{Department of Mathematics, California State University Northridge, Los Angeles, CA; and Department of Computer Science, Tufts University, Medford, MA, USA.
    Email: \texttt{csaba.toth@csun.edu}}
}
\date{}

\usepackage{amssymb,amsthm,amsmath}
\usepackage{graphicx}
\usepackage[small]{caption}
\usepackage{subcaption}
\usepackage{amsfonts}
\usepackage{fullpage}
\usepackage{enumerate}
\usepackage{url}
\usepackage{hyperref}
\usepackage{thm-restate}

\newtheorem{theorem}{Theorem}
\newtheorem{lemma}{Lemma}

\newtheorem{observation}{Observation}
\newtheorem{remark}{Remark}
\newtheorem{corollary}{Corollary}

\newcommand\R{\ensuremath{\mathbb{R}}}
\newcommand\N{\ensuremath{\mathbb{N}}}
\newcommand\eps{\ensuremath{\varepsilon}}

\newcommand\diam{\ensuremath{\mathrm{diam}}}
\newcommand\hull{\ensuremath{\mathrm{hull}}}

\newcommand\across{\ensuremath{A}}
\newcommand\inside{\ensuremath{\mathrm{In}}}
\newcommand\centerset{\inside}

\newcommand\bottomset{\ensuremath{B}}
\newcommand\topset{\ensuremath{T}}
\newcommand\slab{\ensuremath{\mathrm{slab}}}

\newcommand\dist{\ensuremath{\mathrm{dist}}}

\newcommand\rank{\ensuremath{\mathrm{rank}}}
\newcommand\rep{\ensuremath{\mathrm{Rep}}}

\begin{document}
\maketitle

\begin{abstract}
A $t$-spanner of a graph $G=(V,E)$ is a subgraph $H=(V,E')$ that contains a $uv$-path of length at most $t$ for every $uv\in E$. It is known that every $n$-vertex graph admits a $(2k-1)$-spanner with $O(n^{1+1/k})$ edges for $k\geq 1$. This bound is the best possible for $1\leq k\leq 9$ and is conjectured to be optimal due to Erd\H{o}s' girth conjecture.

We study $t$-spanners for $t\in \{2,3\}$ for geometric intersection graphs in the plane. These spanners are also known as \emph{$t$-hop spanners} to emphasize the use of graph-theoretic distances (as opposed to Euclidean distances between the geometric objects or their centers).  We obtain the following results: 
(1) Every $n$-vertex unit disk graph (UDG) admits a 2-hop spanner with $O(n)$ edges; improving upon the previous bound of $O(n\log n)$. 
(2) The intersection graph of $n$ axis-aligned fat rectangles admits a 2-hop spanner with $O(n\log n)$ edges, and this bound is tight up to a factor of $\log \log n$.
(3) The intersection graph of $n$ fat convex bodies in the plane admits a 3-hop spanner with $O(n\log n)$ edges.
(4) The intersection graph of $n$ axis-aligned rectangles admits a 3-hop spanner with $O(n\log^2 n)$ edges.
\end{abstract}

\section{Introduction}
\label{sec:intro}

Graph spanners were introduced by Awerbuch~\cite{Awerbuch85} and by Peleg and Sch\"affer~\cite{peleg1989graph}. A spanner of a graph $G$ is a spanning subgraph $H$ with bounded distortion between graph distances in $G$ and $H$. For an edge-weighted graph $G=(V,E)$, a spanning subgraph $H$ is a $t$-spanner if $d_H(u,v)\leq t\cdot d_G(u,v)$ for all $u,v\in V$, where $d_H$ and $d_G$ are the shortest-path distances in $H$ and $G$, respectively. The parameter $t\geq 1$ is the \emph{stretch factor} of the spanner. A long line of research is devoted to finding spanners with desirable features, which minimize the number of edges, the weight, or the diameter; refer to a recent survey by Ahmed et al.~\cite{AhmedBSHJKS20}. 

In abstract graphs, all edges have unit weight. In a graph $G$ of girth $g$, any proper subgraph $H$ has stretch at least $g-1$. In particular, a complete  bipartite graph does not have any subquadratic size $t$-spanner for $t<3$. The celebrated greedy spanner by Alth\"ofer et al.~\cite{althofer1993sparse} finds, for every $n$-vertex graph and parameter $t=2k-1$, a $t$-spanner with $O(n^{1+\frac{1}{k}})$ edges; and this bound matches the lower bound from the Erd\H{o}s girth conjecture~\cite{Erdos64extremalproblems}. 

\paragraph{Geometric Setting: Euclidean and Metric Spanners.}
Given a set $P$ of $n$ points in a metric space $(M,\delta)$, consider the complete graph $G$ on $P$ where the weight of an edge $uv$ is the distance $\delta(u,v)$. If $M$ has doubling dimension $d$ (e.g.,  Euclidean spaces of constant dimension) the greedy algorithm by Alth\"ofer et al.~\cite{althofer1993sparse} constructs an $(1+\eps)$-spanner with $\eps^{-O(d)}n$ edges~\cite{le2019truly}.
Specifically, every set of $n$ points in $\R^d$ admits a $(1+\eps)$-spanner with $O(\eps^{-d}n)$ edges, and this bound is the best possible~\cite{le2019truly}.

Gao and Zhang~\cite{GaoZ05} considered data structures for approximating the \emph{weighted} distances in \emph{unit disk graphs} (\emph{UDG}), which are intersection graphs of unit disks in $\R^2$. Importantly, the \emph{weight} of an edge is the Euclidean distance between the centers. They designed a well-separated pair-decomposition (WSPD) of size $O(n\log n)$ for an $n$-vertex UDG. 
For the unit ball graphs in doubling dimensions, Eppstein and Khodabandeh~\cite{abs-2106-15234} construct $(1+\eps)$-spanners which also have bounded degree and total weight $O(w(MST))$, generalizing earlier work in $\R^d$ by Damian et al.~\cite{DamianPP06}; see also~\cite{LS-unified1}.
F{\"{u}}rer and Kasiviswanathan~\cite{FurerK12} construct a $(1+\eps)$-spanner with $O(\eps^{-2}n)$ edges for the intersection graph of $n$ disks of arbitrary radii in $\R^2$. 
%

\paragraph{Hop-Spanners for Geometric Intersection Graphs.} 
Unit disk graphs (UDG) were the first geometric intersection graphs for which the hop distance was studied (i.e., the unweighted version), motivated by applications in wireless communication. Spanners in this setting are often called \emph{hop-spanners} to emphasize the use of  graph-theoretic distance (i.e., hop distance), as opposed to the Euclidean distance between centers.


For an $n$-vertex UDG $G$, Yan et al.~\cite{YanXD12} constructed 
a subgraph $H$ with $O(n\log n)$ edges and  $d_H(u,v)\leq 3d_G(u,v)+12$, which is a $15$-hop spanner. 
Catusse et al.~\cite{CatusseCV10} showed that every $n$-vertex UDG admits a 5-hop spanner with at most $10n$ edges (as well as a noncrossing $O(1)$-spanner with $O(n)$ edges). Biniaz~\cite{biniaz2020plane} improved this bound to $9n$. Dumitrescu et al.~\cite{dumitrescu2021sparse} recently showed that every $n$-vertex UDG admits a 5-hop spanner with at most $5.5n$ edges, a 3-hop spanner with at most $11n$ edges, and a 2-hop spanner with $O(n\log n)$ edges. In this paper, we improve the bound on the size of 2-hop spanners to $O(n)$, and initiate the study of minimum 2-hop spanners of other classes of geometric intersection graphs. 

\paragraph{Our Contributions.}
\begin{enumerate}
    \item Every unit disk graph on $n$ vertices admits a 2-hop spanner with $O(n)$ edges (Theorem~\ref{thm:udg} in Section~\ref{sec:udg}). This bound is the best possible, and it generalizes to intersection graphs of translates of a convex body in the plane (Theorem~\ref{thm:translates} in Section~\ref{ssec:translates}).
    \item The intersection graph of $n$ axis-aligned fat rectangles in $\R^2$  admits a 2-hop spanner with $O(n\log n)$ edges (Theorem~\ref{thm:sq} in Section~\ref{sec:sq}). This bound is almost tight: We establish a lower bound of $\Omega(n\log n / \log \log n)$ for the size of 2-hop spanners in the intersection graph of $n$ homothets of any convex body in the plane (Theorem~\ref{thm:lb} in Section~\ref{sec:lb}).
    \item The intersection graph of $n$ fat convex bodies in $\R^2$ admits a 3-hop spanner with $O(n\log n)$ edges    (Theorem~\ref{thm:fat} in Section~\ref{sec:fat}). 
    \item The intersection graph of $n$ axis-aligned rectangles admits a 3-hop spanner with $O(n\log^2 n)$ edges (Theorem~\ref{thm:rectangles} in Section~\ref{sec:rectangles}).
\end{enumerate}

\paragraph{Related Previous Work.}
While our upper bounds are constructive, we do not attempt to minimize the number of edges in a $k$-spanner for a given graph. The \emph{minimum $k$-spanner} problem is to find a $k$-spanner $H$ of a given graph $G$ with the minimum number of edges. This problem is NP-hard~\cite{Cai94,peleg1989graph} for all $2\leq k \leq o(\log n)$;
already for planar graphs~\cite{BrandesH97,Kobayashi18a}. 
It is also hard to approximate up to a factor of $2^{(\log^{1-\eps}n)/k}$, 
for $3\leq \log^{1-2\eps}n$ and $\eps>0$, assuming $NP \not\subseteq BPTIME(2^{\mathrm{polylog}(n)})$~\cite{DinitzKR16}; see also \cite{DodisK99,ElkinP07,Kortsarz01}. 
On the positive side, Peleg and Krtsarz~\cite{KortsarzP98} gave an $O(\log(m/n))$-approximation for the minimum $2$-spanner problem for graphs $G$ with $n$ vertices and $m$ edges; see also~\cite{Censor-HillelD21}. There is an $O(n)$-time algorithm for the minimum 2-spanner problem over graphs of maximum degree at most four~\cite{CaiK94}.

Classical graph optimization problems (which are often hard and hard to approximate) typically admit better approximation ratios or are fixed-parameter tractable (FPT) for geometric intersection graphs. Three main strategies have been developed to take advantage of geometry: (i) Divide-and-conquer strategies using separators and dynamic programming~\cite{An021,BasteT22,BergBKMZ20,CardinalIK21,FominLPSZ19,FominLS12,FoxP14,Lee17};
(ii) Local search algorithms~\cite{BusGMR17,ChanH12,JartouxM22,MustafaR10}; and
(iii) Bounded VC-dimension and the $\eps$-net theory~\cite{AgarwalP20,AronovES10,BusGMR16,PT13,MustafaDG18,mustafa2017epsilon}. It is unclear whether separators and local search help find small $k$-hop spanners. Small hitting sets and $\eps$-nets help finding large cliques in geometric intersection graphs, and this is a tool also used by Dumitrescu et al.~\cite{dumitrescu2021sparse} to construct 2-hop spanners with $O(n \log n)$ edges for UDGs. 

\paragraph{Relation to Edge Clique and Biclique Covers.} 
A 2-hop spanner $H$ of a graph $G=(V,E)$ is union of stars $\mathcal{S}$ such that every edge in $E$ is induced by a star in $\mathcal{S}$. Thus the minimum 2-spanner problem is equivalent to minimizing the sum of sizes of stars in $\mathcal{S}$. As such, the 2-spanner problem is similar to the \emph{minimum dominating set} and \emph{minimum edge-clique cover} problems~\cite{Erdos66EdgeCliqueCover,MichaelQ06}. 
In particular, the size of a 2-hop spanner is bounded above by the minimum \emph{weighted} edge clique cover, where the weight of a clique $K_t$ is $t-1$ (i.e., the size of a spanning star). 
Recently, de Berg et al.~\cite{BergBKMZ20} proposed a divide-and-conquer framework for optimization problems on geometric intersection graphs. Their main technical tool is a weighted separator theorem, where the weight of a separator is $W=\sum_i w(t_i)$ for a decomposition of the subgraph induced by the separator into cliques $K_{t_i}$, and sublinear weights $w(t)=o(t)$. For 2-hop spanners, however, each clique $K_t$ requires a star with $t-1$ edges, so the weight function would be linear $w(t)=t-1$.

Every biclique (i.e., complete bipartite graph) $K_{s,t}$ admits a 3-hop spanner with $s+t-1$ edges (as a union of two stars). Hence an \emph{edge biclique cover}, with total weight $W$ and weight function $w(K_{s,t})=s+t$, yields a 3-hop spanners with at most $W$ edges. Every $n$-vertex graph has an edge biclique cover of weight $O(n^2/\log n)$, and this bound is tight~\cite{ErdosP97,Tuza84}. 
(In contrast, every $n$-vertex graph has a 3-hop spanner with $O(n^{3/2})$ edges~\cite{althofer1993sparse}.) 
Better bounds are known for semi-algebraic graphs, where the edges are defined in terms of semi-algebraic relations of bounded degree. For instance, an incidence graph between $n$ points and $m$ hyperplanes in $\R^d$ admits an edge biclique cover of weight $O((mn)^{1-1/d}+m+n)$~\cite{ApfelbaumS07,BrassK03,SharirS17}.
Recently, Do~\cite{Do19} proved that a semi-algebraic bipartite graph on $m+n$ vertices, where the vertices are points in $\R^{d_1}$ and $\R^{d_2}$, resp., has an edge biclique cover of weight $O_\eps(m^{\frac{d_1d_2-d_2}{d_1d_2-1}+\eps}n^{\frac{d_1d_2-d_1}{d_1d_2-1}+\eps}+m^{1+\eps}+n^{1+\eps})$ for any $\eps>0$. For $d_1+d_2\leq 4$, this result yields nontrivial 3-hop spanners. For a UDG with $m=n$ unit disks,  $d_1=d_2=2$ gives a 3-hop spanner with $W\leq O_\eps(n^{4/3+\eps})$ edges. But for the intersection graph of arbitrary disks in $\R^2$,  $d_1=d_2=3$ gives $O_\eps(n^{3/2+\eps})$, which is worse than the default $O(n^{3/2})$ guaranteed by the greedy algorithm~\cite{althofer1993sparse}.

\paragraph{Representation.} 
Our algorithms assume a geometric representation of a given  intersection graphs (it is NP-hard to recognize UDGs~\cite{BreuK98}, disk graphs~\cite{HlinenyK01,McDiarmidM13}, or box graphs~\cite{Kratochvil94}). Given a set of geometric objects of bounded description complexity, the intersection graph and the hop distances can easily be computed in polynomial time. Chan and Skrepetos~\cite{ChanS19} designed near-quadratic time algorithms to compute all pairwise hop-distances in the intersection graph of $n$ geometric objects (e.g., balls or hyperrectangles in $\R^d$). In a UDG, the hop-distance between a given pair of disks can be computed in optimal $O(n\log n)$ time~\cite{CabelloJ15}.
%
 
 \paragraph{Definitions.} 
 A \emph{body} in Euclidean plane is a compact set with nonempty interior. A set $S\subset \mathbb{R}^2$ is \emph{convex} if for every $a,b\in S$, the line segment $ab$ is contained in $S$. 
The \emph{fatness} of a set $s\subset \R^2$ is the ratio $\varrho_{\mathrm{out}}/\varrho_{\mathrm{in}}$ between the radii of a minimum enclosing disk and a maximum inscribed disk of $s$. A collection $S$ of geometric objects is $\alpha$-fat if the fatness of every $s \in S$ is at most $\alpha$; and it is \emph{fat}, for short, if it is $\alpha$-fat for some $\alpha\in O(1)$. 
An \emph{axis-aliened rectangle} in the plane is the Cartesian product of two closed intervals $R=[a,b]\times [c,d]$, where the \emph{width} of $R$ is $b-a$ and its \emph{height} is $d-c$. In general, the \emph{width} (resp., \emph{height}) of a bounded set $S\subset \mathbb{R}^2$ is the width (resp., height) of its minimum enclosing axis-aligned rectangle. 

\section{Two-Hop Spanners for Unit Disk Graphs}
\label{sec:udg}

In this section, we prove that every $n$-vertex UDG has a 2-hop spanner with $O(n)$ edges. The proof hinges on a key lemma, Lemma~\ref{lem:bipartite}, in a bipartite setting.
A unit disk is a closed disk of unit diameter in $\R^2$; two unit disks intersect if and only if their centers are at distance at most 1 apart. 
%
%
For finite sets $A,B\subset \R^2$, let $U(A,B)$ denote the  unit disk graph on $A \cup B$, and let $G(A,B)$ denote the bipartite subgraph of $U(A,B)$ of all edges between $A$ and $B$.

\begin{restatable}{lemma}{bipartite}
\label{lem:bipartite}
Let $P = A \cup B$ be a set of $n$ points in the plane such that $\diam(A) \le 1$, $\diam(B) \le 1$, and $A$ (resp., $B$) is above (resp., below) the $x$-axis. Then there is a subgraph $H$ of $U(A,B)$ with at most $5n$ edges such that for every edge $ab$ of $G(A,B)$, $H$ contains a path of length at most 2 between $a$ and $b$.
\end{restatable}

We construct the graph $H$ in Lemma~\ref{lem:bipartite} incrementally: In each step, we find a subset $W\subset A\cup B$, together with a subgraph $H(W)$ of at most $5|W|$ edges that contains a $uv$-path of length at most 2 for every edge $uv$ between $u\in W$ and $v\in N(W)$ (cf.\ Lemma~\ref{lem:sparse}); and then recurse on $P\setminus W$. We show that $\bigcup_W H(W)$ is a 2-hop spanner for $U(A,B)$. 

Section~\ref{ssec:hulls} establishes a technical lemma about the interaction pattern of disks in the bipartite setting, which may be of independent interest. One step of the recursion is presented in Section~\ref{ssec:OneStep}. 
The proof of Lemma~\ref{lem:bipartite} is in Section~\ref{ssec:bipartite}. 
Lemma~\ref{lem:bipartite}, combined with previous work~\cite{biniaz2020plane,CatusseCV10,dumitrescu2021sparse} that reduced the problem to a bipartite setting, implies the main result of this section.

\begin{theorem}
\label{thm:udg}
Every $n$-vertex unit disk graph has a $2$-hop spanner with $O(n)$ edges.
\end{theorem}
\begin{proof}
Let $P$ be a set of centers of $n$ unit disks in the plane, and let $G=(P,E)$ be the UDG on $P$. Consider a tiling of the plane with regular hexagons of diameter 1, where each point in $P$ lies in the interior of a tile. A tile $\tau$ is \emph{nonempty} if $\tau\cap P\neq \emptyset$. Clearly $\diam(P\cap \tau)\leq \diam(\tau)= 1$. For each nonempty tile $\tau$, let $S_\tau$ be a spanning star on $P\cap \tau$. 

For each pair of tiles, $\sigma$ and $\tau$, at distance at most $1$ apart, 
Lemma~\ref{lem:bipartite} yields a graph $H_{\sigma,\tau}:=G(A,B)\subset G$ for $A=P\cap \sigma$ and $B=P\cap \tau$ with $5(|P\cap \sigma|+|P\cap \tau|)$ edges. Let $H$ be the union of all stars $S_\tau$ and all graphs $H_{\sigma,\tau}$. It is easily checked that $H$ is a 2-hop spanner of $G$: Indeed, let $uv\in E$. If $u$ and $v$ are in the same tile $\tau$, then $S_\tau$ contains $uv$ or a $uv$-path of length 2. Otherwise $u$ and $v$ are in different tiles, say $\sigma$ and $\tau$, at distance at most 1, and $H_{\sigma,\tau}$ contains $uv$ or a $uv$-path of length 2.

It remains to bound the number of edges in $H$. The union of all stars $S_\tau$ is a spanning forest on $P$, which has at most $n-1$ edges. Every tile $\sigma$ is within unit distance from 18 other tiles~\cite{biniaz2020plane}. The total number of edges in $H_{\sigma,\tau}$ over all pairs of tiles is 
$\sum_{\sigma,\tau} 5(|P\cap \sigma|+|P\cap \tau|) \leq 18 \sum_\sigma 5(|P\cap \sigma|) =90n$. Overall, $H$ has less than $91n$ edges, as required. 
\end{proof}

\subsection{Properties of Unit-Disk Hulls}
\label{ssec:hulls}

Let $A \subset \R^2$ be a finite set of points above the $x$-axis. Let $\mathcal{D}$ be the set of all unit disks with centers on or below the $x$-axis. Let $M(A)$ be the union of all unit disks $D \in \mathcal{D}$ such that $A \cap \mathrm{int}(D) = \emptyset$, and let $\hull(A) = \R^2 \setminus \mathrm{int}(M(A))$; see Fig.~\ref{fig:hull}. The hull of unit disks $\hull(A)$ was introduced in \cite[Section~3]{dumitrescu2021sparse}, as an analogue of  of $\alpha$-shapes~\cite{EdelsbrunnerKS83} and the convex hull of a set $S\subset \R^2$. 

Let us elaborate on the analogy to convex hulls. Recall that the convex hull of a point set $A\subset \R^2$ is the complement of the union of all halfplans that are disjoint from $A$; and it can be constructed by continuously rotating a line tangent to $A$ (e.g., using the gift wrapping algorithm~\cite{Jarvis73}). We can replace halfplanes with unit disks centered on or below the $x$-axis: Imagine that a unit disk $D$ with center on the $x$-axis moves continuously from left to right, but whenever $D$ hits a point $a\in A$, then the center of $D$ rotates counterclockwise around $a$ (and dips below the $x$-axis) until it hits another point in $A$ or returns to the $x$-axis. Such a moving disk $D$ would roll along $\partial \hull(A)$.

\begin{figure}[htbp]
 \centering
 \includegraphics[width=.95\textwidth]{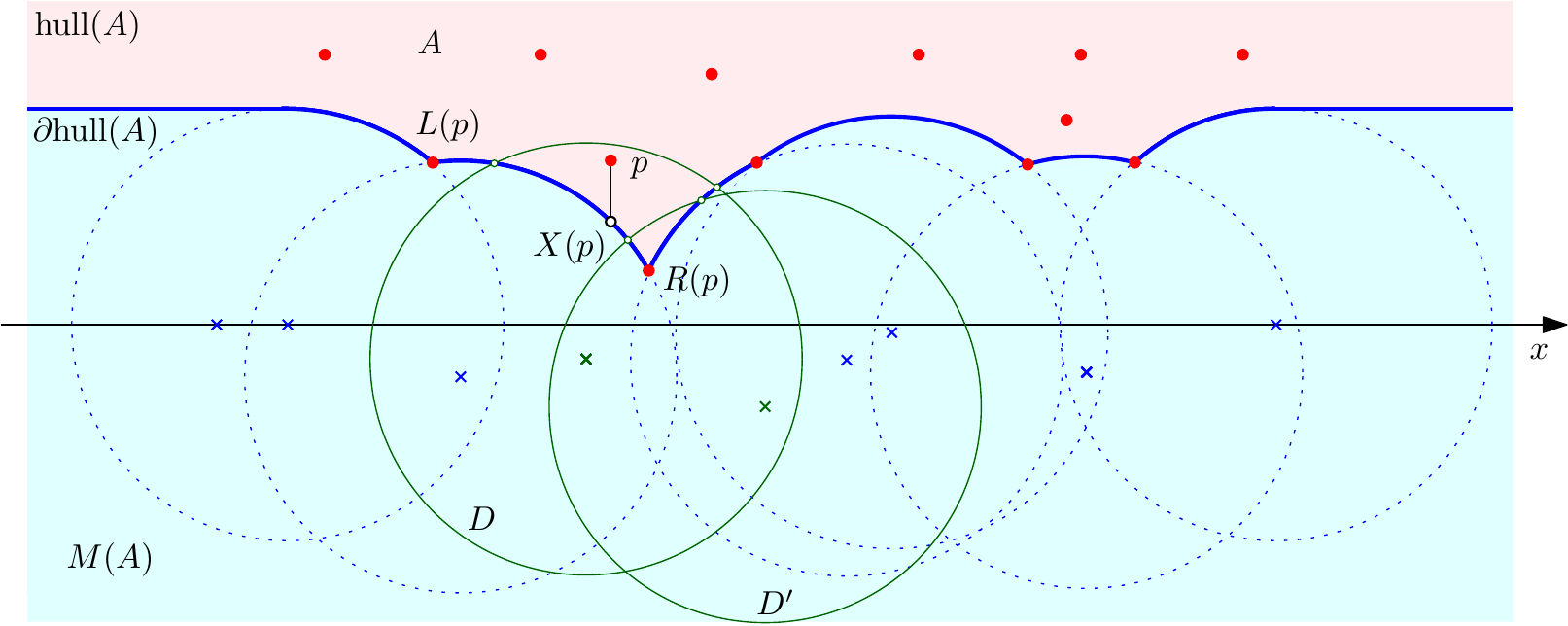}
 \caption{A point set $A$ (red), region $M(A)$ (light blue), and $\hull(A)$ (pink). A point $p\in A$ in a disk $D\in \mathcal{D}$, its vertical projection $X(p)\in \partial \mathrm{hull}(A)$, and the two adjacent points $L(p),R(p)\in A$.}
    \label{fig:hull}
\end{figure}

We show below (cf.\ Lemma~\ref{lem:properties}\eqref{p:i}) that $\partial \hull(A)$ is an $x$-monotone curve. For every $p \in \R^2$ above the $x$-axis, let $X(p)$ denote its vertical projection onto $\partial \hull(A)$; this is well defined because $\partial \hull(A)$ is $x$-monotone. Let $L(p)$ and $R(p)$ denote the points in $A \cap \partial \hull(A)$ immediately to the left and right of $X(p)$ if such a point exists; that is, $L(p)$ (resp., $R(p)$) is the point in $A \cap \partial \hull(A)$ with the largest (resp., smallest) $x$-coordinate that still satisfies $L(p)_x \le X(p)_x$ (resp., $R(p)_x \ge X(p)_x$).

\begin{lemma}
\label{lem:properties}
For every finite set $A \subset \R^2$ above the $x$-axis, the following holds:
\begin{enumerate}
\item\label{p:i} $\partial \hull(A)$ is an $x$-monotone curve.
\item\label{p:ii} For every $D \in \mathcal{D}$, the intersection $D \cap\, \partial \hull(A)$ is connected (possibly empty).
\item\label{p:iii} For every $D \in \mathcal{D}$ and every $p\in A$, if $p\in D$, then $D$ contains $X(p)$. Further, $L(p)$ or $R(p)$ exists, and $D$ contains $L(p)$ or $R(p)$ (possibly both).
\item\label{p:iv} Let $D, D' \in \mathcal{D}$. Suppose that $\partial D$ intersects $\partial \hull(A)$ at points with $x$-coordinates $x_1$ and $x_2$, and $\partial D'$ intersects $\partial \hull(A)$ at points with $x$-coordinates $x_1'$ and $x_2'$. If $x_1 \le x_1' \le x_2' \le x_2$, then $D' \cap \hull(A) \subset D \cap \hull(A)$.
\end{enumerate}
\end{lemma}
%
\begin{proof}
Items \textbf{(1)} and \textbf{(2)} are proven in \cite[Lemma~4]{dumitrescu2021sparse}.

\noindent
\textbf{(3)} As the center of $D$ is below the $x$-axis, if $D$ contains a point $p = (p_x, p_y)$ above the $x$-axis, then $D$ also contains all points $(p_x, p_y')$ where $0 \le p_y' \le p_y$. In particular, if $p \in \hull(A)$, then $p$ is on or above $\partial \hull(A)$, and so $D$ contains $X(p)$.

Let $A_{p}^- \subset A \cap \partial \hull(A)$ be the finite set of points to the left of $X(p)$, and let $A_{p}^+ \subset A \cap \partial \hull(A)$ be the finite set of points to the right of $X(p)$. It was shown in \cite[Lemma~4]{dumitrescu2021sparse} that every disk in $\mathcal{D}$ that contains a point in $A$ also contains some point in $A \cap \partial \hull(A)$. As $D$ contains $p \in A$ by assumption, the intersection $A \cap \partial \hull(A)$ is not empty, so $A_{p}^-$ or $A_{p}^+$ is not empty. Thus, $L(p)$ or $R(p)$ exists.

The intersection $D \cap \partial \hull(A)$ is connected, so if $D$ contains a point in $A_{p-}$ (resp., $A_{p+}$) then it contains $L(p)$ (resp., $R(p)$). As $D$ contains a point in $A \cap \partial \hull(A)$, it also contains a point in $A_{p}^-$ or $A_{p}^+$, and so it contains $L(p)$ or $R(p)$.

\noindent
\textbf{(4)} Let $D, D' \in \mathcal {D}$, and let $C=\partial D$ and $C'=\partial D'$ denote the boundary circles of these disks. As $x_1 \le x_1' \le x_2' \le x_2$, $C$ and $C'$ cross each other an even number of times in $\hull(A)$.
Because $C$ and $C'$ are centered below the $x$-axis, they cross each other at most once above the $x$-axis (cf.~\cite[Lemma~3]{dumitrescu2021sparse}), and so this even number must be zero. Since $C$ and $C'$ do not cross in $\hull(A)$, then $D' \cap \hull(A) \subset D \cap \hull(A)$.
\end{proof}

\subsection{One Incremental Step}
\label{ssec:OneStep}

Let $A$ and $B$ be finite point sets above and below the $x$-axis, respectively, and let $P = A \cup B$. 
For every point $p \in \R^2$, let $N(p) \subset P$ denote the points in $P$ on the opposite side of the $x$-axis within unit distance from $p$; refer to Fig.~\ref{fig:neighbor}. For a point set $S \subset \R^2$, let $N(S) = \bigcup_{p \in S} N(p)$.
Suppose that a unit circle centered at $p\in A$  
intersects $\partial\hull(B)$ at points $p_1, p_2 \in \R^2$; or a unit circle centered at $p\in B$ intersects $\partial\hull(A)$ at points $p_1, p_2 \in \R^2$. Define 
$
I(p) = N(N(p)) \setminus (N(p_1) \cup N(p_2))$;
see Fig.~\ref{fig:neighbor} for an example.

\begin{figure}[htbp]
 \centering
 \includegraphics[width=.95\textwidth]{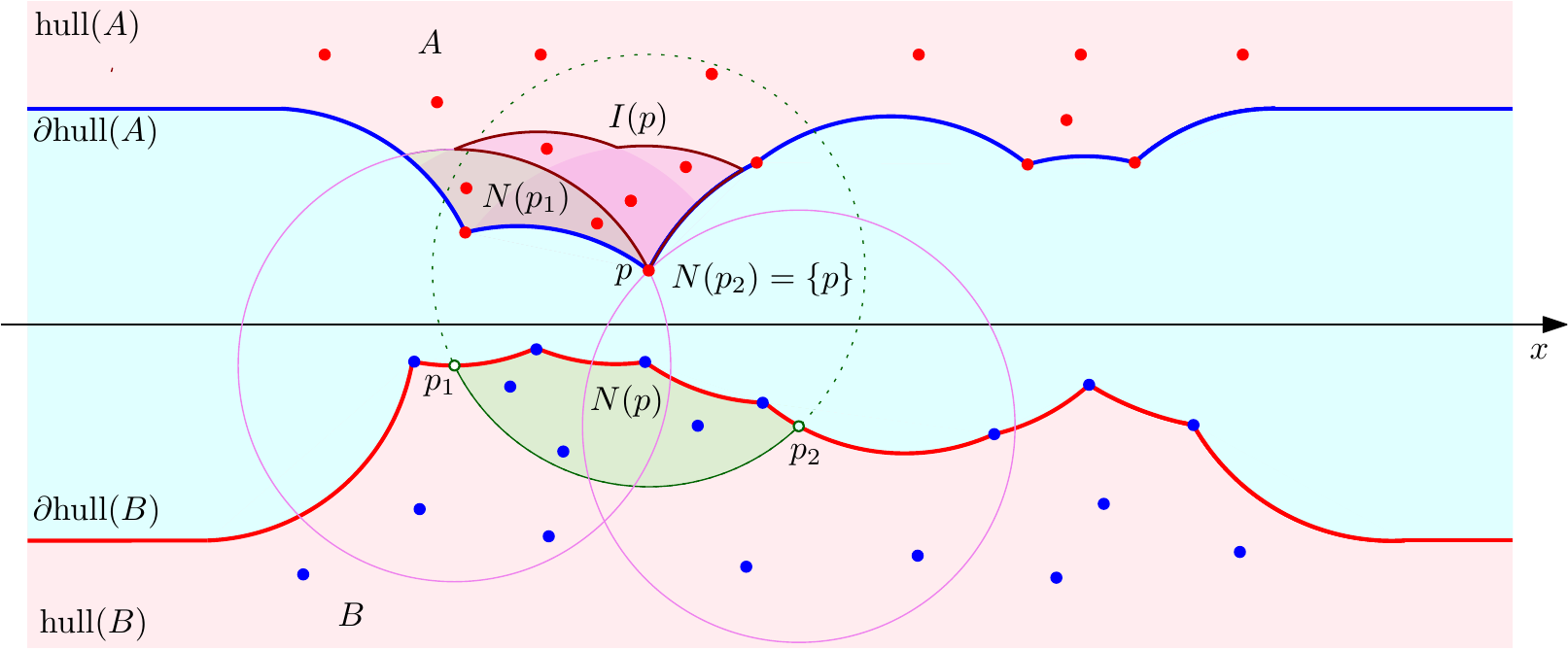}
 \caption{A point $p\in A$ and its neighbors $N(p)\subset B$. 
 The unit circle centered at $p$ intersecting $\partial \hull(B)$ at $p_1$ and $p_2$. The sets $N(p_1)$, $N(p_2)$, and $I(p)$.}
    \label{fig:neighbor}
\end{figure}

\begin{lemma}\label{lem:2ndneighbor}
Let $P = A \cup B$ be a finite set of points in the plane such that $A$ (resp., $B$) is above (resp., below) the $x$-axis. For every $p \in P$, $N(I(p)) \subset N(p)$.
\end{lemma}
\begin{proof}
We may assume w.l.o.g.\ that $p \in A$. Let $v \in I(p)$, and let $D_v$ (resp., $D_p$) denote the unit disk centered on $v$ (resp., $p$). As $v \in N(N(p))$, $D_v$ contains some point $u \in N(p)$. Clearly, $D_p$ contains $u$.
By Lemma~\ref{lem:properties}\eqref{p:iii}, $D_v$ and $D_p$ contain $X(u)$. As $D_p \cap \hull(B)$ has endpoints $p_1$ and $p_2$, Lemma~\ref{lem:properties}\eqref{p:i}--\eqref{p:ii} implies that $X(u)$ has $x$-coordinate between $p_1$ and $p_2$. By definition of $I(P)$, $D_v \cap \hull(B)$ does not contain either $p_1$ or $p_2$, so it only contains points between $p_1$ and $p_2$. By Lemma~\ref{lem:properties}\eqref{p:iv}, $N(v) \subset N(p)$.
\end{proof}

We construct a spanner by repeatedly applying the following lemma:
\begin{lemma}\label{lem:sparse}
Let $P = A \cup B$ be a set of $n$ points in the plane such that $\diam(A) \le 1$, $\diam(B) \le 1$, and $A$ (resp., $B$) is above (resp., below) the $x$-axis. Then there exists a nonempty subset $W \subset P$ and a graph $H(W)$ with the following properties:
\begin{enumerate}
\item $H(W)$ is a subgraph of $U(A, B)$;
\item $H(W)$ contains at most $5 | W |$ edges;
\item for every edge $ab$ with $a \in W$ and $b \in N(W)$, the graph $H(W)$ contains an $ab$-path of length at most 2.
\end{enumerate}
\end{lemma}

\begin{proof}
Let $m \in \R^2$ be a point that maximizes $| N(m) |$ (breaking ties arbitrarily) and let $k = | N(m) |$. Notice that $m$ might not be in $P$. By Lemma~\ref{lem:properties}(\ref{p:iii}), every point in $N(m)$ is within unit distance of $L(m)$ or $R(m)$. Thus, there is some point $v \in \{L(m), R(m)\}$ such that $|N(v)| \ge k/2$. Note that $v \in P$.

Now let $p \in P$ be the point that maximizes $|N(p)|$; and note that $|N(p)|\geq k/2$. Let $W = N(p) \cup I(p) \cup \{p\}$. Let $H(W)$ be the spanning star centered at $p$ connected to all points in $N(N(p))$ and to all points in $N(p)$. We verify that $H(W)$ has the required properties:
\begin{enumerate}
\item Every point in $N(p)$ is within unit distance of $p$. As $p \in A$ and $N(N(p)) \subset A$, every point in $N(N(p))$ is within unit distance of $p$. Thus $H(W)$ is a subgraph of $U(A, B)$. 

\item By definition of $k$, $| N(p_1) | \le k$ and $| N(p_2) | \le k$. Thus, $| N(N(p)) | \le 2k + | I(p) |$. Further, $| W | = | N(p) | + | I(p) | \ge k/2 + | I(p) |$. Thus $| N(N(p)) | \le 4 | W |$. The spanning star $H(W)$ has $| N(N(p)) | + | N(p) | - 1$ edges, so it has at most $5 | W |$ edges.

\item Let $a$ be a point in $W$, and let $b$ be a point in $N(W)$. If $a \in N(p)$, then $b \in N(N(p))$ by the definition of $N(.)$, and so the spanning star contains an $ab$-path of length at most 2. Alternatively, if $a \in I(p)\cup \{p\}$, then $b \in N(p)$ by Lemma~\ref{lem:2ndneighbor}, and so the spanning star contains an $ab$-path of length at most 2.\qedhere
\end{enumerate}
\end{proof}

\subsection{A Bipartite Spanner: Proof of Lemma~\ref{lem:bipartite}}
\label{ssec:bipartite}

We now prove Lemma~\ref{lem:bipartite}, which we restate here for convenience.
\bipartite*
\begin{proof}
Apply Lemma~\ref{lem:sparse} to find a subset $W \subset P$ and a subgraph $H(W)$. Let $H$ be the union of $H(W)$ and the spanner constructed by recursing on $P \setminus W$. Since $H$ is the union of subgraphs of $U(A,B)$, it is itself a subgraph of $U(A, B)$.

\smallskip\noindent\textbf{Stretch analysis.} Suppose $a \in A$ and $b \in B$ are neighbors in $G(A, B)$. We assume w.l.o.g.\ that $a$ was removed before or at the same time as $b$ during the construction of $H$ as part of some subset $W$. Then $H$ includes a subgraph $H(W)$ that, by construction, connects $a$ to all neighbors that have not yet been removed (including $b$) by paths of length at most 2.

\smallskip\noindent\textbf{Sparsity analysis.} Each subgraph $H(W)$ in $H$ is responsible for removing some set of points $W$ and has at most $5 |W|$ edges. Charge 5 edges to each of the $|W|$ points removed. As each point is removed exactly once, $H$ contains at most $5n$ edges.
\end{proof}

\section{Generalization to Translates of Convex Bodies in the Plane}
\label{ssec:translates}

Recently, Aamand et al.~\cite{AamandAKR21} proved that the class of UDGs equals the class of intersection graphs of translates of any strictly convex body in the plane. Hence Theorem~\ref{thm:udg} carries over to this setting. We prove here that Theorem~\ref{thm:udg} generalizes to the intersection graph of a collection of translates of \emph{any} convex body $C\subset \R^2$. Let $G=G(\mathcal{S})$ be the intersection graph of the set $\mathcal{S}=\{C+\mathbf{v}_i: i=1,\ldots , n\}$ of $n$ translates of $C$. We start with a few simplifying assumptions.

\noindent\textbf{(1)} We may assume that $C$ is centrally symmetric. Indeed, assume that $C$ is a convex body in the plane. Two translates of $C$ intersect, say, $C+\mathbf{v}_i \cap C+\mathbf{v}_j\neq \emptyset$, if and only if $\frac12(C-C)+\mathbf{v}_i\cap\frac12(C-C)+\mathbf{v}_j\neq \emptyset$,  where $\frac12(C-C)$ is a centrally symmetric convex body. 
Consequently, the intersection graph of $\mathcal{S}=\{C+\mathbf{v}_i: i=1,\ldots , n\}$ is the same as that of  $\mathcal{S}'=\{\frac12(C-C)+\mathbf{v}_i: i=1,\ldots , n\}$.

\noindent\textbf{(2)} We may assume that $C$ is fat. Indeed, intersections are invariant under nondegenerate affine transformations. By John's ellipsoid theorem~\cite{HarPeled11,John44}, there exists an ellipse $\mathcal{E}\subset \R^2$ such that $\mathcal{E}\subset C\subset 2\, \mathcal{E}$. 
A linear transformation that maps $\mathcal{E}$ to a disk will map $C$ into a $2$-fat centrally symmetric convex body.

In particular, we may assume that $C$ contains a disk of radius $\frac12$ centered at the origin, and is contained in a concentric disk of radius and 1. Thus, if two translates of $C$ intersect, then their centers are within distance 2 apart. The proof of Theorem~\ref{thm:udg} carries over with the only difference that we need to consider pairs of tiles within distance $2$ apart: Every tile is at distance at most 2 from 30 other tiles.

\begin{figure}[htbp]
 \centering
 \includegraphics[width=\textwidth]{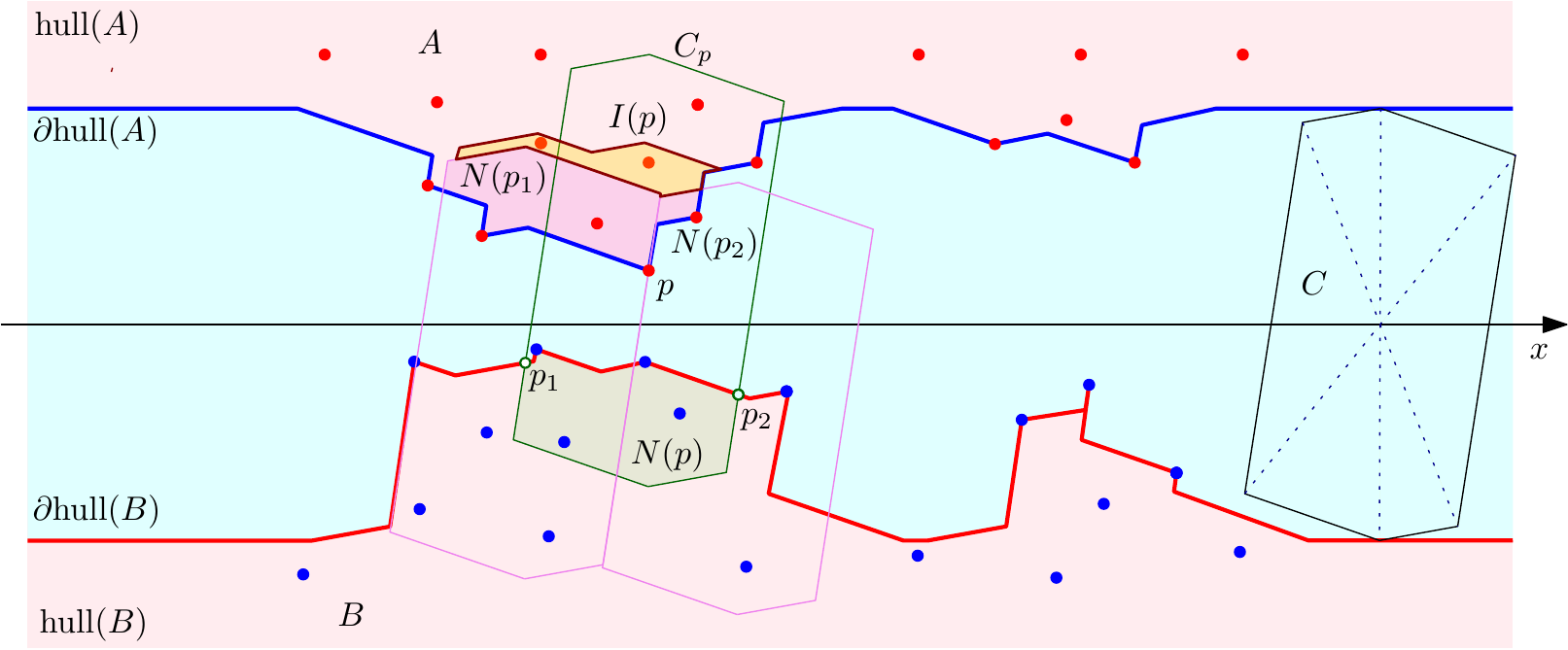}
 \caption{$\hull(A)$ and $\hull(B)$ with respect to a centrally symmetric hexagon $C$. Point $p\in A$ and its neighbors $N(p)\subset B$. The translate $C_p$ of $C$ centered at $p$ intersects $\partial \hull(B)$ at $p_1$ and $p_2$. Sets $N(p_1)$, $N(p_2)$, and $I(p)$.}
    \label{fig:hexneighbor}
\end{figure}

The Lemmata in Section~\ref{ssec:hulls}--\ref{ssec:OneStep} carry over with a few minor adjustments: The set $\hull(A)$ is defined in  terms of translates of $C$ centered on or below the $x$-axis (i.e., unit disk in the Minkowski norm of $C$). Lemma~\ref{lem:properties}\eqref{p:i} holds in a weaker form: $\partial \hull(A)$ is only \emph{weakly $C$-monotone}, that is, it has a connected intersection with every line parallel to $L$, where $L$ is a line tangent to $C$ at an intersection point of $\partial C$ and the $x$-axis (in particular, the $x$-axis partitions $\partial C$ into two $C$-monotone arcs). This allows us to cover $\partial \hull(A)$ with a weakly $C$-monotone curve in $\R^2$. 

In the definition of $I(p)$, the intersections between $\partial \hull(B)$ and the unit circle $\partial C_p$ (in Minkowski-norm) may be line segments. We define $p_1$ and $p_2$ to be the first and last intersection points between the $\partial \hull(B)$ and $\partial C_p$ along the weakly $C$-monotone curve $\partial \hull(B)$. 
With these adjustments, the proof of Lemma~\ref{lem:bipartite} carries over, and we can conclude the main result of this section.

\begin{theorem}\label{thm:translates}
The intersection graph of $n$ translates of a convex body in the plane admits a $2$-hop spanner with $O(n)$ edges.
\end{theorem}

\section{Two-Hop Spanners for Axis-Aligned Fat Rectangles}
\label{sec:sq}

For intersection graphs of $n$ unit disks, we found 2-hop spanners with $O(n)$ edges in Section~\ref{sec:udg}. This bound does not generalize to intersection graphs of disks of arbitrary radii, as we establish a lower bound of $\Omega(n\log n)$ in Section~\ref{sec:lb}. Here, we construct 2-hop spanners with $O(n\log n)$ edges for such graphs under the $L_\infty$ norm (where unit disks are really unit squares). The result also holds for axis-aligned fat rectangles.

We prove a linear upper bound for the 1-dimensional version of the problem (Section~\ref{ssec:interval}), and then address axis-aligned fat rectangles in the plane (Section~\ref{ssec:squares}). 
The \emph{fatness} of a set $s\subset \R^2$ is the ratio $\varrho_{\mathrm{out}}/\varrho_{\mathrm{in}}$ between the radii of a minimum enclosing disk and a maximum inscribed disk of $s$. A collection $S$ of geometric objects is $\alpha$-fat if the fatness of every $s \in S$ is at most $\alpha$; and it is \emph{fat}, for short, if it is $\alpha$-fat for some $\alpha\in O(1)$. 

\subsection{Two-Hop Spanners for Interval Graphs}
\label{ssec:interval}

Let $G(S)$ be the intersection graph of a set $S$ of $n$ closed segments in $\R$. Assume w.l.o.g.\ that $G(S)$ is connected: otherwise, we can apply this construction to each connected component.

We partition $\bigcup{S}$ into a collection of disjoint intervals $\mathcal{I} = \{I_1, \ldots, I_m\}$ as follows.
Let $I_0 = \{p_0\}$ be the interval containing only the leftmost point in $\bigcup S$, and let $k := 1$. While $p_{k-1}$ lies to the left of the rightmost point in $\bigcup S$, let $p_{k}$ be the rightmost point of any segment in $S$ that intersects $p_{k-1}$; let $I_{k} = (p_{k-1}, p_{k}]$; and set $k:=k+1$. As $G(S)$ is connected, this process terminates.
%
For every $k \in \{1, \ldots, m\}$, define the \emph{covering segment} $c_k$ to be some segment that intersects $p_{k-1}$ and has right endpoint $p_k$; see Fig.~\ref{fig:intervals}. Notice that by construction of $I_k$, $c_k$ is guaranteed to exist, and $I_k \subset c_k$.

\begin{figure}[htbp]
 \centering
 \includegraphics[width=.95\textwidth]{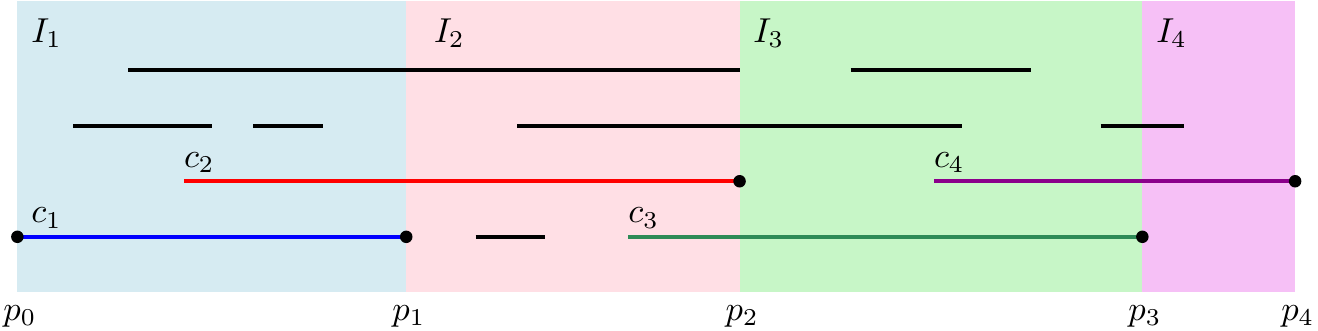}
 \caption{A set of segments $S$, with $\bigcup S$ partitioned into intervals $\mathcal{I} = \{I_1, \ldots, I_4\}$. Each $I_k \in \mathcal{I}$ is contained in some covering segment $c_k \in S$.}
    \label{fig:intervals}
\end{figure}

\begin{lemma}
\label{lem:interval_props}
The set of intervals $\mathcal{I}$ defined above has the following properties:

\begin{enumerate}
\item $\mathcal{I}$ is a partition of $\bigcup{S}$;

\item every segment $s \in S$ intersects at most 2 intervals in $I$;

\item if two segments $a, b \subset \bigcup S$ intersect (with $a, b$ not necessarily elements of $S$), then there is some interval in $I$ that intersects both segments.
\end{enumerate}
\end{lemma}
\begin{proof}
\textbf{(1)} 
By construction, $\bigcup\mathcal{I} = [p_0, p_m]$ where $p_0$ is the leftmost point in $\bigcup S$ and $p_m$ is the rightmost point in $\bigcup S$. As $G(S)$ is connected, $\bigcup S = [p_0, p_m]$. The intervals in $\mathcal{I}$ are pairwise disjoint, so $I$ is a partition of $\bigcup S$.

\textbf{(2)} 
Let $k \in \mathbb{Z}$ be the smallest integer such that $I_k \in \mathcal{I}$ intersects $s$. While $s$ may intersect $I_{k+1} = (p_{k}, p_{k+1}]$, we prove it intersects no other intervals in $I$. Suppose for the sake of contradiction that $s$ intersects some $I_j \in I$ with $j > k + 1$. By construction, all points in $I_j$ lie to the right of $I_{k+1}$, so $s$ contains a point to the right of $p_{k+1}$. Further, because $s$ is connected and $s$ contains a point in $I_k = (p_{k-1}, p_k]$, we have $p_k \in s$. This contradicts the greedy construction for $p_{k+1}$: $p_{k+1}$ is by definition the rightmost point out of all segments in $S$ that intersect $p_{k}$.

\textbf{(3)} As $a$ and $b$ intersect, there exists some point $p$ such that $p \in a \cap b$. By property (1), there is some interval that contains $p$. This interval intersects both $a$ and $b$ (at $p$).
\end{proof}

\begin{theorem}\label{thm:intervals}
Every $n$-vertex interval graph admits a 2-hop spanner with at most $2n$ edges.
\end{theorem}
\begin{proof}
We construct the 2-hop spanner $H$ as the union of stars. For every interval $I_k \in I$, construct a star $H_k$ centered on the covering segment $c_k$ with an edge to every segment that intersects $I_k$. As $I_k \subset c_k$, every segment that intersects $I_k$ also intersects $c_k$, so there is an edge between the two segments in $G(S)$. Define $H = \bigcup_{k = 1}^m H_k$.

\smallskip\noindent\textbf{Stretch analysis.}
Suppose $s_1, s_2 \in S$ intersect. By Lemma~\ref{lem:interval_props}(3),  $s_1\cap s_2$ intersects some interval $I_k$. Thus, the star $H_k\subset H$  connects $s_1$ and $s_2$ by a path of length at most 2.

\smallskip\noindent\textbf{Sparsity analysis.}
Suppose the star $H_k \subset H$ has $j$ edges. The corresponding interval $I_k \in \mathcal{I}$ intersects $j + 1$ segments in $S$. Charge 1 edge to each of the segments intersecting $I_k$. By Lemma~\ref{lem:interval_props}(2), each of the $n$ segments in $S$ is charged at most twice.
\end{proof}

\begin{corollary}
\label{cor_2d_line}
The intersection graph of a set of $n$ axis-aligned rectangles in $\R^2$ that all intersect a fixed horizontal or vertical line admits a 2-hop spanner with at most $2n$ edges.
\end{corollary}

\subsection{Extension to Axis-Aligned Fat Rectangles}
\label{ssec:squares}

Let $G(S)$ be the intersection graph of a set $S$ of $n$ axis-aligned $\alpha$-fat closed rectangles in the plane. For every pair of intersecting rectangles $a, b\in S$, select some representative point in $a \cap b$. Let $\rep(S)$ denote the set comprising the representatives for all intersections.

\paragraph{Setup for a Divide \& Conquer Strategy.}
We recursively partition the plane into slabs by splitting along horizontal lines. The recursion tree $\mathcal{P}$ is a binary tree, where each node $P \in \mathcal{P}$ stores a slab, denoted $\slab(P)$, that is bounded by horizontal lines $b_P$ and $t_P$ on the bottom and top, respectively. The node $P$ also stores a subset $S(P) \subset S$ of (not necessarily all) rectangles in $S$ that intersect $\slab(P)$. 

Let the \emph{inside set} $\inside(P) \subset S(P)$ be the set of rectangles contained in $\mathrm{int}(\slab(P))$. Let the \emph{bottom set} $\bottomset(P) \subset S(P)$ be rectangles that intersect the line $b_P$, the \emph{top set} $\topset(P) \subset S(P)$ be the rectangles that intersect $t_P$, and the \emph{across set} $\across(P) = \bottomset(P) \cap \topset(P)$; see Fig.~\ref{fig:slab}.

\begin{figure}[htbp]
 \centering
 \includegraphics[width=.95\textwidth]{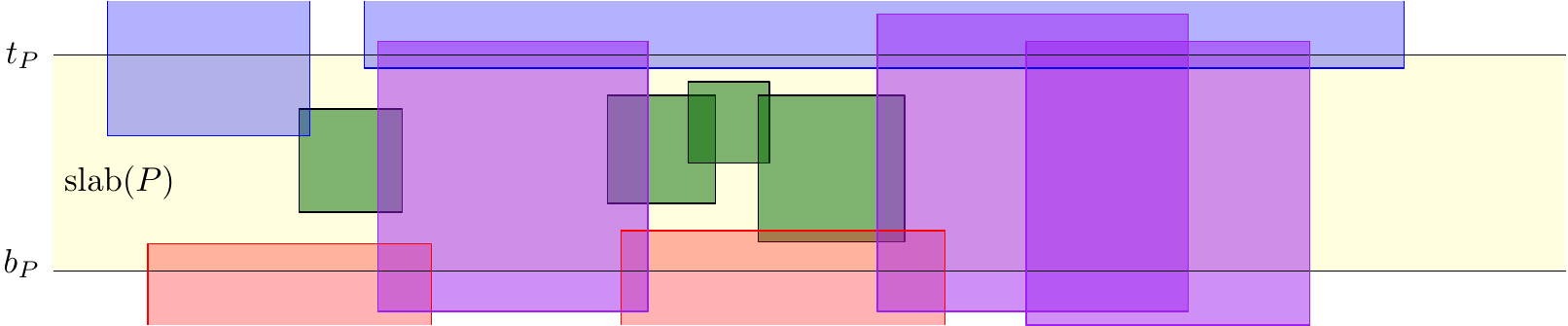}
 \caption{A horizontal $\slab(P)$ is bounded by $b_P$ and $t_P$. 
 Rectangles in the inside set $\inside(P)$ (green), bottom set $\bottomset (P)$ (red and purple), top set $\topset(P)$ (blue and purple), and across set $\across(P)=\bottomset(P)\cap \topset(P)$ (purple). 
 Some red and blue fat rectangles are shown only partially, in a small neighborhood of $\slab(P)$.}
    \label{fig:slab}
\end{figure}

We define the root node $P_r$ to have a slab large enough to contain all rectangles in $S$, and define $S(P_r) = S$. We define the rest of the space partition tree recursively. Let $P \in \mathcal{P}$. Define $\rep(P) \subset \rep(S)$ to be the set $\rep(S) \cap \mathrm{int}(\slab(P))$. If $\rep(P)=\emptyset$, then $P$ is a leaf and has no children. Otherwise, $P$ has two children $P_1$ and $P_2$. Let $c_P$ be a horizontal line with at most half the points in $\rep(P)$ on either side. Let $\slab(P_1)$ (resp., $\slab(P_2)$) be the slab bounded by $b_P$ and $c_P$ (resp., $c_P$ and $t_P$); and let $S(P_1) \subset S(P) \setminus \across(P)$ (resp., $S(P_2)\subset S(P) \setminus \across(P)$) be the set of rectangles that intersect this slab, excluding rectangles in $\across(P)$. Notice that no rectangles in $\across(P)$ appear in the children of $P$, whereas rectangles in the sets $\inside(P)$, $\bottomset(P) \setminus \across(P)$, and $\topset(P) \setminus \across(P)$ appear in one or both of the children.

\paragraph{Spanner Construction.}
We construct a spanner $H(S)$ for $G(S)$ as the union of subgraphs $H(P)$ for each node $P$ in the space partition tree.


We construct $H(P)$ such that there is a path of length at most 2 between every rectangle $s \in \across(P)$ and every rectangle in $S(P)$ that $s$ intersects. Every edge in $G(S(P))$ requiring such a path involves a rectangle in $\bottomset(P)$, a rectangle in $\topset(P)$, or a rectangle in $\inside(P)$. We construct three subgraphs to deal with these three categories of edges.

By Corollary~\ref{cor_2d_line}, we can construct a subgraph $H_{\bottomset}(P)$ of $G(\bottomset(P))$ with at most $2\, |\bottomset(P)|$ edges that is a 2-hop spanner for $\bottomset(P)$. Similarly, we can construct a 2-hop spanner $H_{\topset}(P)$ for $G(\topset(P))$ with $2\, |\topset(P)|$ edges. As all rectangles in $\across(S)$ intersect $b_P$, we can apply Corollary~\ref{cor_2d_line} to construct a 2-hop spanner $H_{\centerset}'(P)$ with at most $2\, |\across(P)|$ edges for $G(\across(P))$. To construct $H_{\inside}(P)$, we partition $\bigcup \big(\across(P) \cap \slab(P)\big)$ analogously to the 1-dimensional case.

Recall that by Lemma~\ref{lem:interval_props}(1), the line segment $\bigcup \big(\across(P) \cap b_P\big)$ can be partitioned into intervals $I_k$, each of which is contained in some covering segment $c_k \in \across(P)$. As every $s \in \across(P)$ is an axis-aligned rectangle that spans two horizontal lines $b_P$ and $t_P$, the segments in $I_k$ can be extended upward to form axis-aligned rectangles $\widehat{I}_k$, each with an associated \emph{covering rectangle} $\widehat{c}_k \in \across(P)$ corresponding to the covering segment $c_k$ in the 1-dimensional case. Let $\widehat{\mathcal{I}}$ denote the set of all 2-dimensional intervals $\widehat{I_k}$.

We construct $H_{\inside}(P)$ from $H_{\inside}'(P)$ using these intervals. For every $s \in \inside(P)$, if $s$ intersects some $\widehat{I}_k \in \widehat{\mathcal{I}}$, add an edge between $s$ and $\widehat{c}_k$ to $H_{\inside}'(P)$. Let $H(P) = H_{\bottomset}(P) \cup H_{\topset}(P) \cup H_{\inside}(P)$.

\paragraph{Stretch and Weight Analysis.}
We start with a technical lemma (Lemma~\ref{lem:sequencewidth}), which is used in the stretch and weight analysis for the graph $H(P)$ of a single node $P\in \mathcal{P}$ (Lemma~\ref{lem:subproblemspanner}).
Notice that the intervals in $\widehat{\mathcal{I}}$ act similarly to the 1-dimensional intervals in $\mathcal{I}$: in particular, Lemma~\ref{lem:interval_props} carries over, with $I_k$ replaced by $\widehat{I}_k$, and with the line segment $\bigcup S$ replaced by the region $\bigcup \big(A(P) \cap \slab(P)\big)$.

\begin{lemma}
\label{lem:sequencewidth}
Let $w$ denote the minimum width of any rectangle in $A(P)$. Then for any $k \in \mathbb{N}$, the union of any $2k$ contiguous intervals in $\widehat{\mathcal{I}}$ has width at least $k w$.
\end{lemma}
\begin{proof}
By construction, every covering rectangle $\widehat{c}_k$ intersects $\widehat{I}_k$ and $\widehat{I}_{k-1}$. By Lemma~\ref{lem:interval_props}(2), $\widehat{c}_k$ does not intersect any other intervals in $\widehat{I}$. Thus, $\widehat{c}_k \subset \widehat{I}_{k-1} \cup \widehat{I}_{k}$. This means that every pair of intervals has width at least $w$. As there are $k$ disjoint pairs of intervals in a set containing $2k$ contiguous intervals, such a set must have width at least $kw$.
\end{proof}

\begin{lemma}\label{lem:subproblemspanner}
The subgraph $H(P)$ has the following properties:
\begin{enumerate}
    \item for every edge $ab \in G(P)$ with $a \in \across(P)$, 
     $H(P)$ contains an $ab$-path of length at most 2; 
    \item $H(P)$ contains $O(\alpha^2\, |S(P)|)$ edges.
\end{enumerate}
\end{lemma}
\begin{proof}
\textbf{(1)} Every $b \in S(P)$ is in $\bottomset(P)$, $\topset(P)$, or $\inside(P)$. If $b \in \bottomset(P)$, then the claim follows from the definition of $H_{\bottomset}(P)$ and the fact that $H_{\bottomset}(P)$ is a subgraph of $H(P)$. Similarly, the claim holds when $b \in \topset(P)$.

Suppose $b \in \inside(P)$. By Lemma~\ref{lem:interval_props}(3), if there is an edge $ab$ in $G(P)$ then both $a$ and $b$ intersect some interval $\widehat{I}_k$. By construction of $H_{\centerset}(P)$, there is an edge between $a$ and $c_k$ and between $b$ and $c_k$ (or else either $a$ or $b$ is equal to $c_k$) and so there is a path of length at most 2 between $a$ and $b$ in $H_{\centerset}(P)$. As $H_{\centerset}(P)$ is a subgraph of $H(P)$, this proves the claim.

\smallskip \noindent \textbf{(2)} 
By construction, $H_{\bottomset}(P)$ contains $2\, | \bottomset(P) |$ edges, $H_{\bottomset}(P)$ contains $2\, | \topset(P) |$ edges, and $H_{\centerset}'(P)$ contains $2\, | \across(P) |$ edges.

We now bound the number of edges that are added to $H_{\centerset}'(P)$ to produce $H_{\centerset}(P)$. Let $h$ be the distance between $b_P$ and $t_P$. Every rectangle $a \in \across(P)$ has width at least $\Omega(\frac{h}{\alpha})$, as $a$ is $\alpha$-fat and has height at least $h$. Further, every rectangle $b \in \inside(P)$ has width less than $\alpha h$, as otherwise $b$ would cross $b_P$ or $t_P$.

Let $\widehat{I}_l, \widehat{I}_r \in \widehat{I}$, resp., be the leftmost and rightmost intervals that the rectangle $b \in \inside(P)$ intersects. If $\widehat{I}_l = \widehat{I}_r$, then $b$ intersects only one interval in $\widehat{\mathcal{I}}$.  Suppose this is not the case. Then $\widehat{I}_l$ and $\widehat{I}_r$ are interior-disjoint. The union of the intervals of $\widehat{\mathcal{I}}$ between $\widehat{I}_l$ and $\widehat{I}_r$ (if any exist) has width less than $\alpha h$; otherwise, $b$ could not intersect both $\widehat{I}_l$ and $\widehat{I}_r$. By Lemma~\ref{lem:sequencewidth}, the union of any consecutive $2\alpha^2$ intervals (each of width at least $h/\alpha$) has width at least $\alpha h$. Thus, $b$ intersects at most $2\alpha^2 - 1$ intervals other than $\widehat{I}_l$ and $\widehat{I}_r$.

This implies that every $b \in \inside(P)$ adds at most $O(\alpha^2)$ edges to $H_{\centerset}'$ during the construction of $H_{\centerset}$. Thus, $H_{\centerset}$ has at most $2 | \across(P) | + \alpha^2 | \inside (P) |$ edges. As $\inside(P)$, $\bottomset(P)$, $\topset(P)$, and $\across(P)$ are all subsets of $S(P)$, $H(P)$ has at most $O(\alpha^2 |S(P)|)$ edges.
\end{proof}

We prove that $H(S) = \bigcup_{P \in \mathcal{P}} H(P)$ has $O(\alpha^2 n \log n)$ edges and that it is a 2-hop spanner. We begin by considering the size. While some $H(P)$ may contain many edges, we bound the total size of $H(S)$ by showing that every rectangle in $S$ is involved in $O(\log n)$ subproblems.

\begin{lemma}
\label{lem:square_subproblem_counts}
For every rectangle $s \in S$, the following hold:
\begin{enumerate}
    \item there are $O(\log n)$ nodes $P \in \mathcal{P}$ where $s \in \inside(P)$;
    \item there are $O(\log n)$ nodes $P \in \mathcal{P}$ where $s \in \bottomset(P) \setminus \across(P)$; symmetrically, there are $O(\log n)$ nodes $P \in \mathcal{P}$ where $s \in \topset(P) \setminus \across(P)$;
    \item there are $O(\log n)$ nodes $P \in \mathcal{P}$ where $s \in \across(P)$.
\end{enumerate}
\end{lemma}
\begin{proof}
Notice that for any $k$, the slabs of nodes at level $k$ in the space partition tree have pairwise disjoint interiors. Since $S$ contains $n$ rectangles, there are at most $\binom{n}{2}$ intersections in $G(S)$. Thus, $| \rep(S) | \leq \binom{n}{2}$, and so the tree has $O(\log n)$ levels.

\textbf{(1)} For every level $k\in \mathbb{N}$ in the space partition tree, there is only one node $P$ where $s \in \inside(P)$. Suppose for the sake of contradiction that $s \in \inside(P_1)$ and $s \in \inside(P_2)$ with $P_1$ and $P_2$ in the same level and $P_1 \neq P_2$. By the definition of $\inside(.)$, $s$ is contained in $\slab(P_1)$ and in $\slab(P_2)$. As these slabs are disjoint, this is impossible. Summation over $O(\log n)$ levels of the recursion tree completes the proof.

\textbf{(2)} For every level $k \in \mathbb{N}$ in the tree, consider the node $P$ with the highest slab such that $\slab (P) \cap s \neq \emptyset$. Notice that $s \in \bottomset(P)$ and $s \notin \topset(P)$, so $s \in \bottomset(P) \setminus \across(P)$. Any other node $P'$ in this level that $s$ intersects lies strictly below $P$ (as nodes within a level have pairwise disjoint slab interiors) and $s$ is connected, so $s \in \bottomset(P')$ only if $s \in \topset(P')$. Thus, $P$ is the only node in level $k$ where $s \in \bottomset(P) \setminus \across(P)$. A symmetric argument proves that there is only one $P$ per level where $s \in \topset(P) \setminus \across(P)$.

\textbf{(3)} For every level $k \in \N$ in the tree, there are at most two nodes $P$ such that $s \in \across(P)$. Suppose for the sake of contradiction that there exist distinct $P_{1}$, $P_{2}$, and $P_{3}$ at level $k$ such that $s \in \across(P_{1}) \cap \across(P_{2}) \cap \across(P_{3})$. The interiors of the corresponding slabs are disjoint, so we may assume w.l.o.g.\ that $P_{1}$ lies below $P_{2}$, which lies below $P_{3}$.
As $s$ is connected, it intersects $b_P$ and $t_P$ for every node $P$ between $P_{1}$ and $P_{3}$. In particular, $s$ must be in $\across(P)$ for the sibling $P$ of $P_{2}$. Then $s$ is also in $\across(P')$ for the parent $P'$ of $P_2$. This is a contradiction---if $s$ were in the $\across$ set of the parent of $P_{2}$, it would not have been added to the set $S(P_{2}) \subset S(P') \setminus A(P')$ of rectangles for the child.
\end{proof}

\begin{corollary}
\label{cor:size}
For every $s \in S$, there are $O(\log n)$ nodes $P \in \mathcal{P}$ where $s \in S(P)$.
\end{corollary}
\begin{proof}
This follows from the fact that for every node $P$, $S(P)$ is the union of the four sets mentioned in Lemma~\ref{lem:square_subproblem_counts}: $S(P) = \inside(P) \cup \left(\bottomset(P) \setminus \across(P)\right) \cup \left(\topset(P) \setminus \across(P)\right) \cup \across(P)$. 
\end{proof}

\begin{lemma}
\label{lem:efficiency_axis}
$H(S)$ has $O(\alpha^2 n \log n)$ edges.
\end{lemma}
\begin{proof}
For every node $P$, $H(P)$ has $O(\alpha^2\, | S(P) |)$ edges by Lemma~\ref{lem:subproblemspanner}. Charge $O(\alpha^2 )$ edges to each rectangle in $S(P)$. By Corollary~\ref{cor:size}, each rectangle is charged at most $O(\log n)$ times, and so $H(S)$ has at most $O(\alpha^2 n \log n)$ edges.
\end{proof}

\begin{lemma}
\label{lem:correctness_axis}
$H(S)$ is a 2-hop spanner for $G(S)$.
\end{lemma}
\begin{proof}
Let $ab$ be an edge in $G(S)$. As the rectangles $a$ and $b$ intersect, there is some point $p \in \rep(S)$ that lies in $a \cap b$. 
Since $p$ is not in the interior of any slab at the leaf level, a horizontal line of the space partition contains $p$. 
Assume w.l.o.g. that this line is $b_P$ for some node $P$.
If both $a$ and $b$ are present in $S(P)$, then $H_B(P)$ contains an $ab$-path of length at most 2. 
Otherwise, there is some node $P'$ for which both $a$ and $b$ are in $S(P')$ but either $a$ or $b$ is \emph{not} in the set for either child of $P'$. Assume w.l.o.g.\ that $a$ was removed. By construction, a rectangle is removed exactly when it is in $\across(P')$. By Lemma~\ref{lem:subproblemspanner}, $H(P')$ contains an $ab$-path of length at most 2. As $H(S) = \bigcup_{P \in \mathcal{P}} H(P)$, this proves that $H(S)$ contains such a path.
\end{proof}

The previous two lemmata prove the following theorem.

\begin{theorem}\label{thm:sq}
The intersection graph of every set of $n$ axis-aligned $\alpha$-fat rectangles in the plane admits a 2-hop spanner with $O(\alpha^2 n\log n)$ edges.
\end{theorem}

\section{Three-Hop Spanners for Fat Convex Bodies}
\label{sec:fat}

In this section, we generalize Theorem~\ref{thm:sq} to fat convex bodies in the plane, at the expense of increasing the hop-stretch factor from 2 to 3. 
A convex body $C\subset \R^d$ is \emph{$\alpha$-fat} if the ratio between the radii of the minimum enclosing and maximum inscribed balls is bounded by $\alpha$. 
A family of convex bodies are called \emph{fat} if they are $\alpha$-fat for some constant $\alpha$.

\subsection{Spanner Construction}
\label{ssec:fat-construction}

Let $S$ be a set of fat convex bodies in the plane. For each $s\in S$, let $r(S)$ be the axis-parallel bounding box of $s$, which is also fat (possibly with a larger $\alpha$); and let $\mathcal{R}=\{r(s):s\in S\}$.

We partition $\mathbb{R}^2$ into slabs w.r.t.\ $\mathcal{R}(S)$ as described in Section~\ref{ssec:squares}. Notice that $s \in S$ intersects a horizontal line if and only if $r(s) \in \mathcal{R}(S)$ intersects that line. Let $S(P)$, $\bottomset(P)$, $\topset(P)$, $\across(P)$, and $\inside(P)$ denote subsets of $S$ (not $\mathcal{R}(S)$) following the same criteria as in Section~\ref{ssec:squares}. For a leaf node $P$, we will use the same spanner $H(P)$; for an internal node $P$, we modify $H(P)$ to capture all intersections involving $\across(P)$ in the slab. As before, let $H_{\bottomset}(P)$ (resp., $H_{\topset}(P)$) be a spanner for the intersection graph of $\{s\cap b_P: s\in \bottomset(P)\}$ (resp., $\{s\cap t_P: s\in \topset(P)\}$), using Theorem~\ref{thm:intervals}. Unlike in the case of axis-aligned fat rectangles, $H_{\bottomset}(P)$ is \textit{not} necessarily a spanner for all edges between $\bottomset(P)$: it is possible for two convex bodies to intersect in $\slab(P)$ but not intersect along $b_P$ (see, for example, $l$ and $r$ in Fig.~\ref{fig:fatbodies}). We modify $H_{\inside}(P)$ to account for these intersections, as well as to handle intersections between $\across(P)$ and $\inside(P)$ for fat convex bodies.

\paragraph{Construction of $H_{\inside}'(P)$.}
We construct a set of spanning stars with centers $\mathcal{C}=\{c_0,\ldots , c_m\}$. This is similar to the construction for rectangles. However, in the construction for rectangles we chose $c_{k+1}$ greedily from the neighbors of $c_k$, while now we choose $c_{k+1}$ from the neighbors of a larger set. 
In our construction, $c_{k}$ and $c_{k+1}$ do not necessarily intersect. 
We construct another set of centers $\mathcal{C}'=\{c_0',\ldots , c_{m-1}'\}$
to cover all bodies between $c_k$ and $c_{k+1}$. 

For the construction of $H_{\inside}'$, we only consider intersections within $\slab(P)$: $a_1$ and $a_2$ \emph{intersect in $\slab(P)$} if $a_1\cap a_2\cap\slab(P)\neq\emptyset$; and they are \emph{disjoint w.r.t.\ $\slab(P)$} if $a_1\cap a_2\cap\slab(P)=\emptyset$.
We assume w.l.o.g.\ that the intersection graph of $\across(P)$ w.r.t.\ $\slab(P)$ is connected; otherwise, we repeat this construction for each component.
Sort the objects $a\in \across(P)$ in increasing order according to the minimum $x$-coordinate of $a\cap b_P$, and break ties arbitrarily. For all $a\in \across(P)$, let $\rank(a)$ denote the rank of $a$ in this order; see Fig.~\ref{fig:fatrank}.

We construct $\mathcal{C}$ and $\mathcal{C}'$ as follows. 
Let $c_0 \in \across(P)$ be the object with smallest rank, and $k:=0$. 
While $\rank(c_k)<|\across(P)|$, let $c_{k+1}$ be the body in $\across(P)$ with maximum rank that intersects in $\slab(P)$ some body with rank less than or equal to $\rank(c_k)$; and set $k:=k+1$.
Let $L_k =\{a\in A(P): \rank(a)\leq \rank(c_k)\}$. 
For every $k\in \{0,\ldots, m-1\}$, let $c_k'\in \across(P)$ be the body in $L_k$ of maximal rank that intersects $c_{k+1}$ in $\slab(P)$.
Let $\mathcal{C} = \{c_0, \ldots, c_m\}$ and let $\mathcal{C}' = \{c_0', \ldots, c_{m-1}'\}$.

We can now construct $H_{\inside}'(P)$. For every $c \in \mathcal{C} \cup \mathcal{C}'$, add an edge between $c$ and every $a \in A(P)$ that it intersects in $\slab(P)$. Let $H_{\inside}'(P)$ be the union of these stars.

\begin{figure}[htbp]
 \centering
 \includegraphics[width=.95\textwidth]{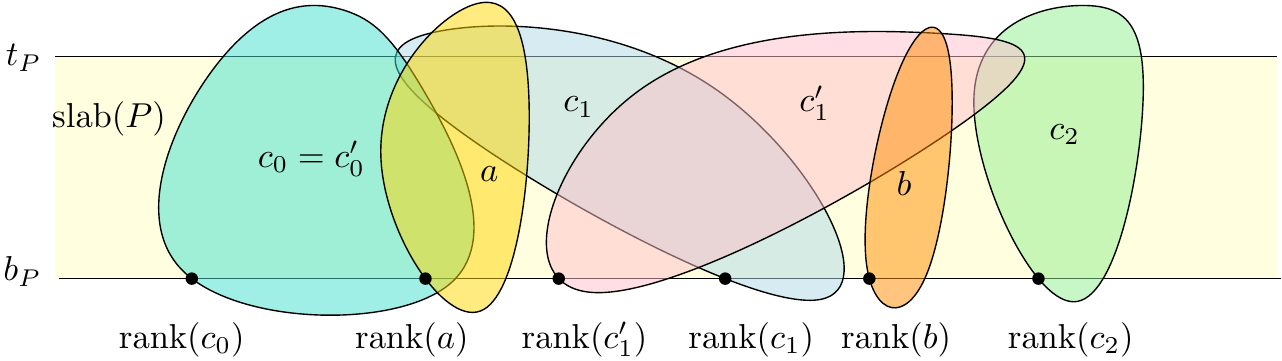}
 \caption{Fat convex bodies in $\across(P)$, ordered by rank. Some bodies are in $\mathcal{C}$ and $\mathcal{C}'$.}
    \label{fig:fatrank}
\end{figure}

\paragraph{Construction of $H(P)$.}
We add edges to $H_{\inside}'(P)$. For every $c \in \mathcal{C} \cup \mathcal{C}'$, let $\mathcal{N}(c) \subset \across(P)$ denote the neighbors of $c$ in $H_{\inside}'(P)$. For every $s \in \inside(P)$, add an edge from $s$ to every $c \in \mathcal{C}' \cup \mathcal{C}$ that it intersects. If $s$ intersects some body in $\mathcal{N}(c)$ but not $c$, add an edge from $s$ to one such body in $\mathcal{N}(c)$. Repeat this process for every $s \in \bottomset(P)$ (resp., $\topset(P)$), considering only the convex body $a \in A(P)$ where $s \cap a$ is contained in the interior of $\slab(P)$ and does not intersect $b_P$ (resp., $t_P$). Call the resulting graph $H_{\inside}(P)$. Let $H(P) = H_{\bottomset}(P) \cup H_{\topset}(P) \cup H_{\inside}(P)$.

\subsection{Stretch Analysis}
\label{ssec:fat-stretch}
In the rectangular case, there was a simple proof for the stretch factor based on the fact that $\slab(P)$ could be partitioned into intervals, each covered by some rectangle in $S(P)$. Thus every pair of intersecting rectangles were attached to a common star. This is no longer the case here. 
Instead, we show that every convex body $a\in \across(P)$ intersects a nearby star center in $\mathcal{C}$ or $\mathcal{C}'$, using the greedy choice of successive elements in $\mathcal{C}$ and $\mathcal{C}'$. We cannot guarantee that every pair of intersecting convex bodies are in a common star, but if they are not, they will be in two nearby stars with adjacent centers. 
%

\begin{lemma}\label{lem:order}
The following properties hold for every $P$:
\begin{enumerate}
    \item\label{o:i} if $l, m, r \in A(P)$ with $\rank(l) \leq \rank(m) \leq \rank(r)$ and $l$ intersects $r$ in $\slab(P)$, then $m$ intersects $l$ or $r$ in $\slab(P)$;
    \item\label{o:ii} for every $k \in \{0, \hdots, m-1\}$, $c_{k}' \in \mathcal{C}'$ intersects $c_{k} \in \mathcal{C}$ in $\slab(P)$;
    \item\label{o:iii} every $a \in \across(P)$ intersects some $c \in \mathcal{C} \cup \mathcal{C}'$ in $\slab(P)$.
\end{enumerate}
\end{lemma}

\begin{figure}[htbp]
 \centering
 \includegraphics[width=.8\textwidth]{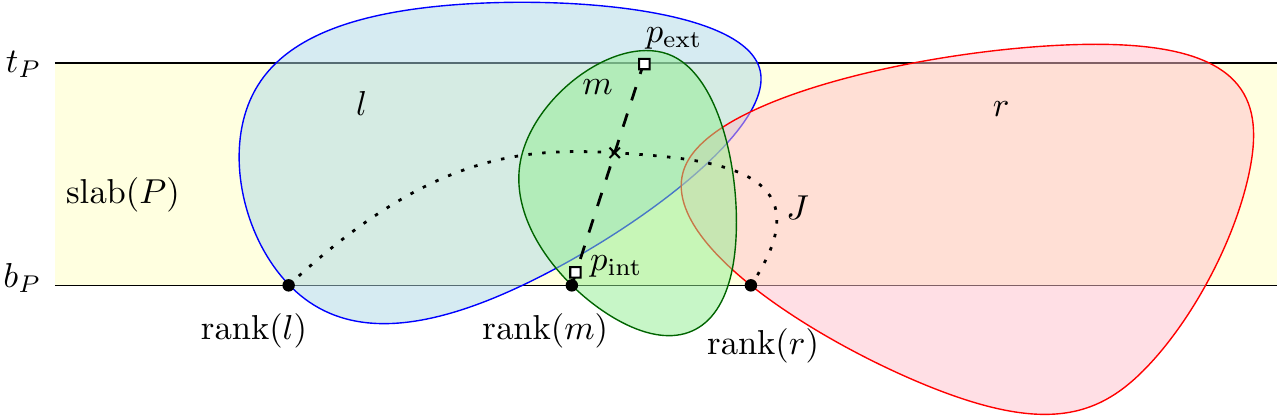}
 \caption{Fat convex bodies $l, m, r \in A(P)$ with $\rank(l) < \rank(m) < \rank(r)$. The curve $J$, and the line segment between $p_{\rm ext}$ and $p_{\rm int}$, from the proof of Lemma~\ref{lem:order}\eqref{o:i}.}
    \label{fig:fatbodies}
\end{figure}

\begin{proof}
\textbf{(1)} As $l$ and $r$ intersect in $\slab(P)$, there is a Jordan arc $J$  in $(l \cup r) \cap \slab(P)$ from the leftmost point in $l \cap b_P$ to the leftmost point in $r \cap b_P$; see Fig.~\ref{fig:fatbodies}. By connecting the endpoints of $J$ along the bottom line $b_P$, we obtain a closed Jordan curve $J_0$. As $m \in \across(P)$, the body $m$ contains some point $p_\text{ext}$ in $m \cap t_P$. Note that $p_\text{ext}$ is either in the exterior of $J_0$ or on the curve $J$. As $\rank(l) \leq \rank(m) \leq \rank(r)$, the body $m$ contains some point $p_\text{int}$ that is in the interior of $J_0$: Take $p_\text{int}$ to be a point slightly above the leftmost point in $m \cap b_P$, so that $p_\text{int}$ lies in the interior of $\slab(P)$ instead of on $b_P$.
By the Jordan curve theorem, the line segment between $p_\text{int}$ and $p_\text{ext}$ intersects $J_0$. As this segment does not intersect $b_P$, it intersects $J$. By convexity, $m$ intersects $J$.

\textbf{(2)} Suppose for the sake of contradiction that $c_{k}'$ did not intersect $c_{k}$ in $\slab(P)$. By construction, $\rank(c_{k}') < \rank(c_{k}) < \rank(c_{k+1})$. As $c_{k}'$ intersects $c_{k+1}$, property (\ref{o:i}) implies that $c_{k}$ intersects $c_{k}'$ or $c_{k+1}$. The latter case is impossible, as $c_{k}'$ by construction has the highest rank of all bodies in $L_k$ that intersect $c_{k+1}$, and $c_k$ has higher rank than $c_{k}'$.

\textbf{(3)} We may assume w.l.o.g.\ that $a \neq c_m$, so there exists some $k \in \mathbb{Z}$ such that $\rank(c_{k}) \leq \rank(a) < \rank(c_{k + 1})$, which exists as $a \neq c_m$. As $\rank(c_{k}') \leq \rank(c_{k})$, property (\ref{o:i}) implies that $a$ intersects either $c_{k}'$ or $c_{k+1}$ in $\slab(P)$.
\end{proof}

\begin{lemma}\label{lem:across-3hop}
If $a, b \in A(P)$ intersect in $\slab(P)$, then $H_{\inside}'(P)$ contains an $ab$-path of length at most $3$.
\end{lemma}
\begin{proof}
We may assume w.l.o.g.\ that $\rank(a) < \rank(b)$. If $a$ or $b$ is in $\mathcal{C} \cup \mathcal{C}'$, then $H_{\inside}'(P)$ contains the edge $ab$. Otherwise, there exists some $k \in \mathbb{Z}$ such that $\rank(c_{k}) < \rank(b) < \rank(c_{k+1})$. For simplicity, we assume that $c_{k}' \neq c_{k}$ and that $c_{k - 1}$ exists. If this is not the case, simpler versions of the same arguments carry over. We consider four cases; see Fig.~\ref{fig:across-3hop}.

\begin{figure}[htbp]
 \centering
 \begin{subfigure}{.45\textwidth}
  \centering
  \includegraphics[width=\textwidth]{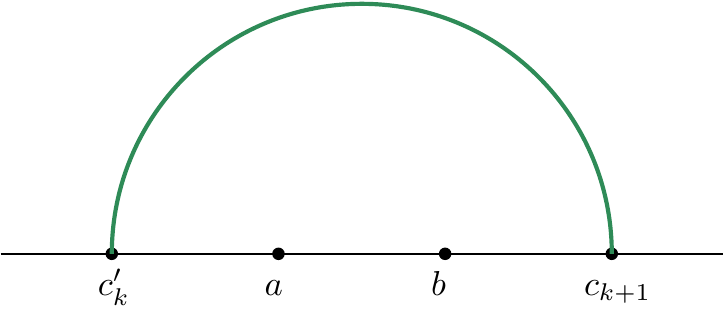}
  \caption{Case 1}
 \end{subfigure}
 \begin{subfigure}{.45\textwidth}
  \centering
  \includegraphics[width=\textwidth]{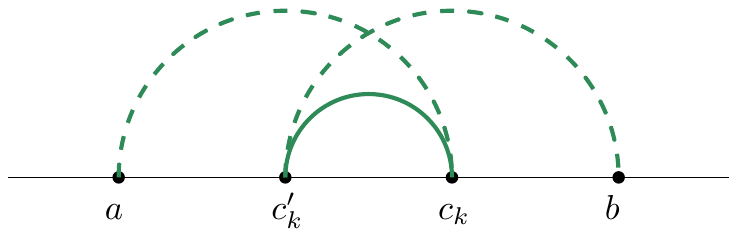}
  \caption{Case 2}
 \end{subfigure}
 \begin{subfigure}{.45\textwidth}
  \centering
  \includegraphics[width=\textwidth]{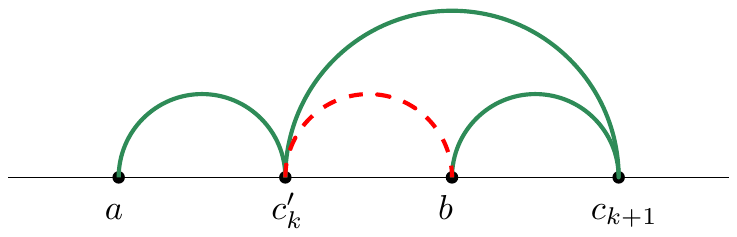}
  \caption{Case 3}
 \end{subfigure}
 \begin{subfigure}{.45\textwidth}
  \centering
  \includegraphics[width=\textwidth]{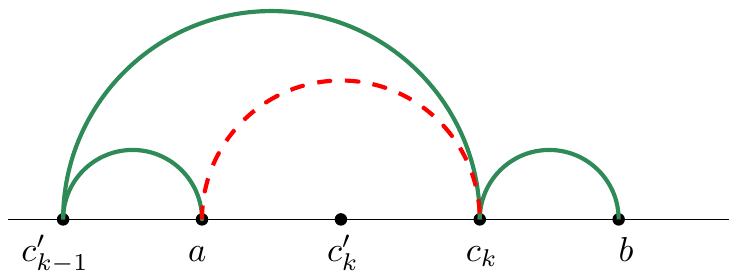}
  \caption{Case 4}
 \end{subfigure}
 \caption{Four cases considered in the proof of Lemma~\ref{lem:across-3hop}. Each body is represented by a point on a line, ordered by $\rank$. Green edges indicate that two bodies intersect, and that this intersection is used in a $3$-hop path from $a$ to $b$. Red edges indicate no intersection exists. Dashed edges indicate that an intersection exists (or not) by assumption.
 }
\label{fig:across-3hop}
\end{figure}

\smallskip\noindent\textbf{Case~1: $\rank(c_{k}') < \rank(a)$.} 
By construction, $\rank(c_{k}') \leq \rank(c_{k})$, and so $\rank(c_{k}') < \rank(a) < \rank(b) < \rank(c_{k+1})$. By Lemma~\ref{lem:order}(\ref{o:i}), $a$ (resp., $b$) intersects $c_a \in \{c_{k}', c_{k+1}\}$ (resp., $c_b \in \{c_{k}', c_{k+1}\}$). If $c_a = c_b$ then $H_{\inside}'(P)$ contains the $2$-hop path $(a, c_a, b)$; otherwise, $H_{\inside}'(P)$ contains a $3$-hop path $(a, c_a, c_b, b)$.

\smallskip\noindent\textbf{Case~2: $\rank(c_{k}') > \rank(a)$, $b$ intersects $c_{k}'$, and $a$ intersects $c_{k}$.}
By Lemma~\ref{lem:order}(\ref{o:ii}), $c_{k}$ intersects $c_{k}'$. Thus, $H_{\inside}'(P)$ contains the $3$-hop path $(a, c_{k}, c_{k}', b)$.

\smallskip\noindent\textbf{Case~3: $\rank(c_{k}') > \rank(a)$ and $b$ does not intersect $c_{k}'$.} Lemma~\ref{lem:order}(\ref{o:i}) implies that $b$ intersects $c_{k}'$ or $c_{k+1}$, and that $c_{k}'$ intersects $a$ or $b$. By assumption, $b$ does not intersect $c_{k}'$, so $b$ intersects $c_{k+1}$ and $a$ intersects $c_{k}'$. Thus, $H_{\inside}'(P)$ contains the $3$-hop path $(a, c_{k}', c_{k+1}, b)$.

\smallskip\noindent\textbf{Case~4: $\rank(c_{k}') > \rank(a)$ and $a$ does not intersect $c_{k}$.}
As $\rank(c_{k}') \leq \rank(c_{k})$, we have $\rank(a) < \rank(c_{k})$. Further, $\rank(c_{k-1}') < \rank(a)$: Otherwise, $a \in L_{k}$ and $a$ intersects $b$ with $\rank(b) > \rank(c_k)$, which contradicts the greedy construction of $\mathcal{C}$. Thus, $\rank(c_{k-1}') < \rank(a) < \rank(c_{k})$, and Lemma~\ref{lem:order}(\ref{o:i}) implies that $a$ intersects $c_{k-1}'$ or $c_{k}$. By assumption, $a$ does not intersect $c_{k}$, so $a$ intersects $c_{k-1}'$. Further, $\rank(a) < \rank(c_{k}) < \rank(b)$, so Lemma~\ref{lem:order}(\ref{o:i}) implies that $c_{k}$ intersects $b$. We conclude that $H_{\inside}'(P)$ contains the $3$-hop path $(a, c_{k-1}', c_{k}, b)$; see also Fig.~\ref{fig:fatspanner} for an example.
\end{proof}

\begin{lemma}\label{lem:fat-3hop}
If $a \in \across(P)$ and $b \in S(P)$ intersect in $\slab(P)$, then $H(P)$ contains an $ab$-path of length at most 3.
\end{lemma}
\begin{proof}
If $b \in \across(P)$, this follows from Lemma~\ref{lem:across-3hop} and the fact that $H_{\inside}'(P)$ is a subgraph of $H(P)$. If $b$ intersects $a$ in $b_P$ or $t_P$, then this follows from the fact that $H_{\bottomset}(P)$ and $H_{\topset}(P)$ are subgraphs of $H(P)$. Otherwise, by Lemma~\ref{lem:order}(\ref{o:iii}) there is some $c \in \mathcal{C}' \cup \mathcal{C}$ such that $a = c$ or $a \in \mathcal{N}(c)$. By construction, if $b$ intersects $c$ in $\slab(P)$, then $H_{\inside}(P)$ contains the $2$-hop path $(b, c, a)$. Otherwise, there is some $n \in \mathcal{N}(c)$ that $b$ intersects in $\slab(P)$. If $n = a$, then $H_{\inside}(P)$ contains the edge $ab$. Otherwise, $H_{\inside}(P)$ contains the $3$-hop path $(b, n, c, a)$; see Fig.~\ref{fig:fatspanner} for an example.
\end{proof}

\begin{figure}[htbp]
 \centering
\includegraphics[width=.95\textwidth]{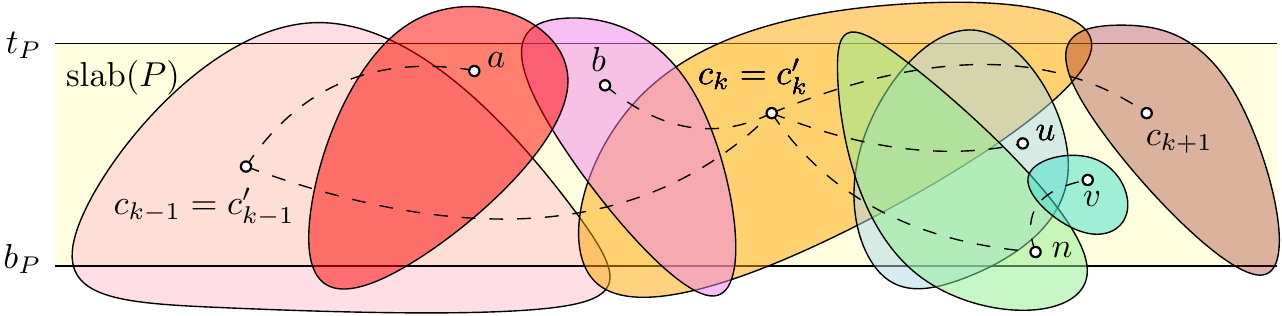}
 \caption{A set $S(P)$ of fat convex bodies, and the spanner $H(P)$ with edges marked by dashed lines. There is a $3$-hop path between $a \in \across(P)$ and $b \in \across(P)$ (cf. case 4 in Lemma~\ref{lem:across-3hop}). There is a $3$-hop path between $u \in \across(P)$ and $v \in \inside(P)$ (cf.\ Lemma~\ref{lem:fat-3hop})}
\label{fig:fatspanner}
\end{figure}

Analogously to Lemma~\ref{lem:correctness_axis}, this immediately implies the following.
\begin{corollary}\label{cor:correctness_axis}
$H(S)$ is a 3-hop spanner for $G(S)$.
\end{corollary}

\subsection{Analyzing the Number of Edges}
\label{ssec:fat-size}

We show that the number of edges in $H(P)$ is $O(\alpha^3\, |S(P)|)$.  Analogously to Lemmata~\ref{lem:square_subproblem_counts}--\ref{lem:efficiency_axis}, this implies that $H(S)=\bigcup_{P} H(P)$ has $O(\alpha^3\, n\log n)$ edges. We first consider $H'_{\inside}(P)$ and $H_{\inside}(P)$. To bound the numnber of edges in $H'_{\inside}(P)$, we show that every $a \in A(P)$ can only intersect star centers in $\mathcal{C}$ and $\mathcal{C}'$ with rank close to $\rank(a)$. For convenience, we define $\rank(c_i)=-1$ for all $i<0$ (even though $c_i$ does not exist). 

\begin{lemma}\label{lem:order-intersect}
The following properties hold for every $P$:
\begin{enumerate}
    \item\label{p:order-i} for every $c_j' \in \mathcal{C}'$ with $j \geq 1$, $\rank(c_{j-1}) < \rank(c_{j}') \leq \rank(c_{j})$; and
    \item\label{p:order-ii} for every $a \in A(P)$, let $k$ denote the integer such that $\rank(c_{k-1}) < \rank(a) \leq \rank(c_{k})$.
    If $a$ intersects some $b \in A(P)$ in $\slab(P)$, then $\rank(c_{k-2}) < \rank(b) \leq \rank(c_{k+1})$.
\end{enumerate}
\end{lemma}
\begin{proof}
\textbf{(1)} $\rank(c_{j}') \leq \rank(c_{j})$ holds by construction. Suppose for the sake of contradiction that $\rank(c_{j-1}) \geq \rank(c_{j}')$. As $c_{j}' \in L_{j-1}$ and intersects $c_{j+1}$, the greedy construction of $\mathcal{C}$ implies that $\rank(c_{j}) \geq \rank(c_{j+1})$, which is a contradiction. 

\textbf{(2)}
Suppose for contradiction that $\rank(b) > \rank(c_{k+1})$. As $a \in L_{k}$ and intersects $b$, the greedy construction implies that $\rank(b) \leq \rank(c_{k+1})$, a contradiction.
Similarly, suppose $\rank(b) \leq \rank(c_{k-2})$. As $b \in L_{k-2}$ and $b$ intersects $a$, we have $\rank(a) \leq \rank(c_{k-1})$ by the greedy construction of $\mathcal{C}$. As $\rank(a) > \rank(c_{k-1})$ by assumption, this is a contradiction.
\end{proof}

\begin{lemma}\label{lem:Hin-fat}
$H_{\inside}'(P)$ has $O(|A(P)|)$ edges.
\end{lemma}
\begin{proof}
Let $a \in A(P)$ and let $k$ denote the integer such that $\rank(c_{k-1}) < \rank(a) \leq \rank(c_k)$. Notice that for every $j \geq k+2$, $\rank(c_j) \ge \rank(c_{k+2}) > \rank(c_{k+1})$; by Lemma~\ref{lem:order-intersect}(\ref{p:order-ii}), $a$ does not intersect $c_j$. Similarly, for every $i \leq k - 2$, $\rank(c_{i}) \leq \rank(c_{k-2})$ and so $a$ does not intersect $c_{i}$. It follows that $a$ can only intersect 3 star centers in $\mathcal{C}$: $c_{k-1}$, $c_k$, and $c_{k+1}$.

Further, Lemma~\ref{lem:order-intersect}(\ref{p:order-i}) implies that for every $j \geq k+2$, $\rank(c_{j}') \geq \rank(c_{k+2}') > \rank(c_{k+1})$ and that for every $i \leq k-2$, $\rank(c_{i}') \leq \rank(c_{k-2})$. It follows from Lemma~\ref{lem:order-intersect}(\ref{p:order-ii}) that $a$ intersects at most $3$ star centers in $\mathcal{C}'$: $c_{k-1}'$, $c_k'$, and $c_{k+1}'$. As every $k \in \across(P)$ contributes at most $6 = O(1)$ edges to $H_{\inside}'(P)$, in total $H_{\inside}'(P)$ has $O(|\across(P)|)$ edges.
\end{proof}

\begin{lemma}\label{lem:disjoint}
We can partition $\mathcal{C}\cup\mathcal{C}'$ into four sets that are each disjoint w.r.t.\ $\slab(P)$. 
\end{lemma}
\begin{proof}
By Lemma~\ref{lem:order-intersect}(\ref{p:order-ii}), every $c_k\in \mathcal{C}$ is disjoint from all $c_j\in \mathcal{C}$ with $j\leq k-2$. Consequently, the bodies in $\mathcal{C}$ with even (resp., odd) indices are pairwise disjoint. 

Similarly, let $c_k' \in \mathcal{C}'$, and let $c_j' \in \mathcal{C}$ with $j \leq k - 2$. By Lemma~\ref{lem:order-intersect}(\ref{p:order-i}) we have $\rank(c_{k-1}) < \rank(c_k') \leq \rank(c_k)$, and $\rank(c'_j) \leq \rank(c_j) \leq \rank(c_{k-2})$.
It follows from Lemma~\ref{lem:order-intersect}(\ref{p:order-ii}) that $c_j'$ is disjoint from $c_k'$.
Consequently, the bodies in $\mathcal{C}'=\{c_0',\ldots , c_{m-1}'\}$ with even (resp., odd) indices are pairwise disjoint. 
\end{proof}

Our next objective is to bound the number of edges in $H_{\inside}(P)$. We give a geometric argument that every $s \in \inside(P)$ can intersect no more than $O(1)$ disjoint bodies in $A(P)$ (this is accomplished in Lemma~\ref{lem:fat-inside} below). The \emph{height} of $\slab(P)$, denoted $h(P)$, is the distance between $b_P$ and $t_P$. 

\begin{lemma}\label{lem:inscribed}
For every $a\in \across(P)$, the intersection $a\cap \slab(P)$ contains a 
disk $d$ with $\diam(d)\geq \Omega(h(P) / \alpha)$. 
\end{lemma}
\begin{proof}
Since $a\in \across(P)$, it intersects both $b_P$ and $t_P$, and so $\diam(a)\geq h(P)$. Let $d_{\rm in}$ be a maximal inscribed disk of $c$, with radius $\varrho_{\rm in}$ and center $c_{\rm in}$; and let $\varrho_{\rm out}$ be the radius of a minimum enclosing disk of $a$. Since $a$ is $\alpha$-fat, then $\diam(d_{\rm in}) = 2\,\varrho_{\rm in}\geq 2\,\varrho_{\rm out}/\alpha\geq \diam(a)/\alpha\geq h(P)/\alpha$.

Let $m_P$ be the horizontal symmetry axis of $\slab(P)$, which is at distance $\frac12\,h(P)$ from both $b_p$ and $t_P$. We may assume w.l.o.g.\ that the center of $d_{\rm in}$ is on or above $m_P$. We distinguish between three cases based on the location of $p_{\rm in}$; see Fig.~\ref{fig:inscribed}. 

\begin{figure}[htbp]
 \centering
\includegraphics[width=.8\textwidth]{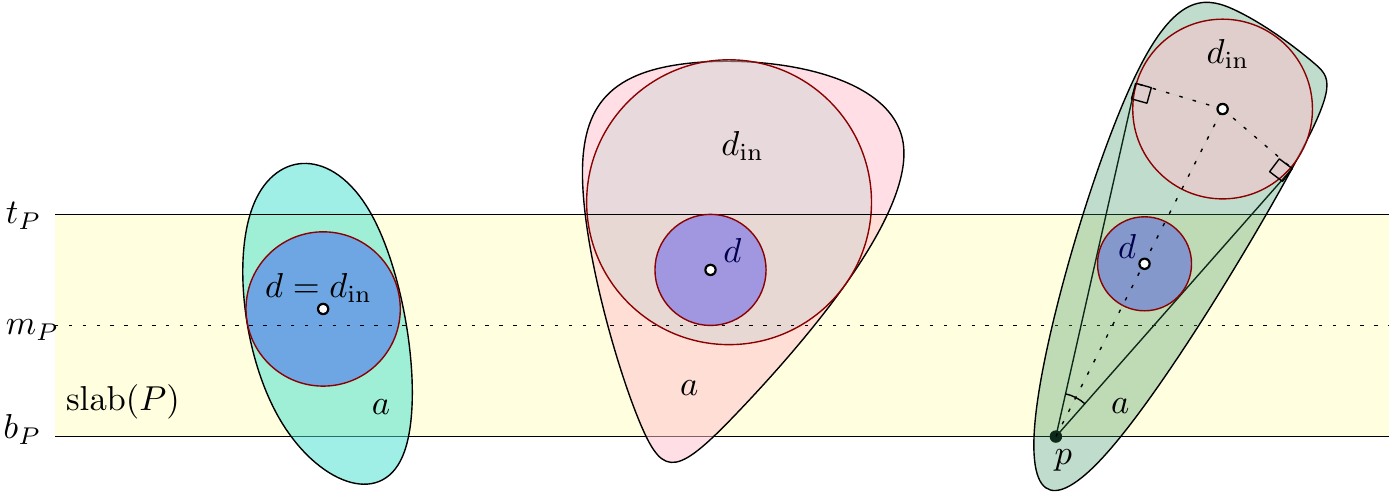}
 \caption{Three cases considered in the proof of Lemma~\ref{lem:inscribed}. 
 Left: $d_{\rm in}\subset \slab(P)$. 
 Middle: $d_{\rm in}$ intersects both $m_P$ and $t_P$. 
 Right: $d_{\rm in}$ is above $m_P$.}
    \label{fig:inscribed}
\end{figure}

\noindent \textbf{(1) $d_{\rm in}\subset \slab(P)$.}
Then we can take $d=d_{\rm in}$, and then $\diam(d)=\diam(d_{\rm in})\geq \Omega(h(P) / \alpha)$.

\noindent\textbf{(2) $d_{\rm in}$ intersects both $m_P$ and $t_P$.}
Then $d_{\rm in}$ contains a disk $d$ tangent to both $m_p$ and $t_P$. Necessarily, $d\subset d_{\rm in}\subset a$, and $\diam(d)=\frac12\, h(P)\geq \Omega(h(P)/\alpha)$.

\noindent\textbf{(3) $d_{\rm in}$ is above $m_P$.}
Let $p\in b_P\cap a$. Since $a$ contains both $c_{\rm in}$ and $a$, then $\|c_{\rm in} p\|\leq \diam(a)$. The angle between the two tangent lines from $p$ to $d_{\rm in}$ is  
\begin{equation}\label{eq:cp1}
2\, \sin^{-1} \left(\frac{\varrho_{\rm in}}{\|c_{\rm in}p\|}\right) 
\geq 2\,\sin^{-1}\left(\frac{\diam(d_{\rm in})/2}{\diam(a)}\right) 
\geq 2\,\sin^{-1}\left(\frac{1}{2\alpha}\right) 
\geq \Omega\left(\frac{1}{\alpha}\right).
\end{equation}
Dilate $d_{\rm in}$ from center $p$ to obtain a disk $d\subset \slab(P)$ tangent to $t_P$.
By construction, $d$ is contained in the convex hull of $d_{\rm in}$ and $p$, which is contained in $a$ by convexity. Then $\diam(d)\geq \Omega(h(P)/\alpha)$, as required.
\end{proof}

\begin{figure}[htbp]
 \centering
 \includegraphics[width=\textwidth]{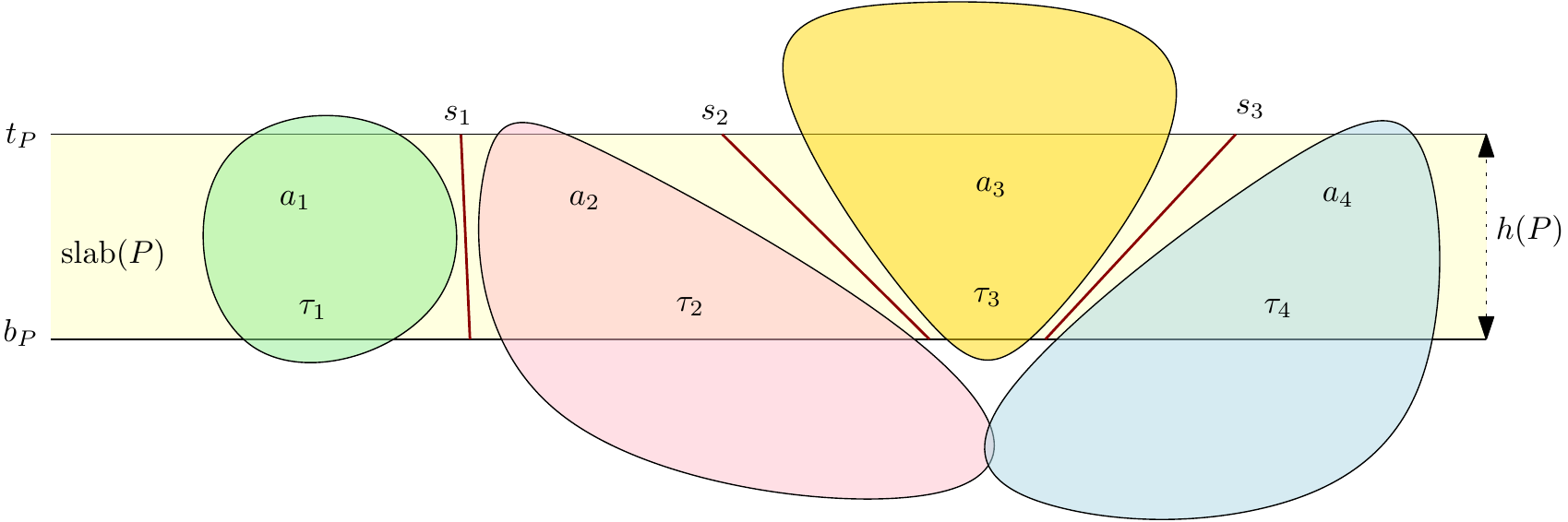}
 \caption{Fat convex bodies $a_1,\ldots ,a_4 \in A(P)$ that do not intersects in $\slab(P)$. Their intersections with $\slab(P)$ are enclosed in interior-disjoint trapezoids $\tau_1,\ldots , \tau_4$.}
    \label{fig:trapezoids}
\end{figure}

Let $a_1, \ldots, a_m\in \across(P)$ be a sequence of convex bodies such that 
their intersections with $\slab(P)$ are pairwise disjoint and $\rank(a_1)<\ldots  <\rank(a_m)$. By convexity, $a_i\cap \slab(P)$ and $a_{i+1}\cap \slab(P)$ can be separated by line segments $s_i$, $i\in \{1,\ldots ,m-1\}$, between $b_P$ and $t_P$. These segments define interior-disjoint trapezoids $\tau_1, \ldots, \tau_m\subset \slab(P)$ such that $a_i\cap \slab(P)\subset \tau_i$ for $i=1,\ldots ,m$; see Fig~\ref{fig:trapezoids} for an illustration.

We show that any two consecutive segments, $s_i$ and $s_{i+1}$, are far apart or are far from being parallel. Specifically, let $\dist(s_i,s_{i+1})$ be the minimum distance between a point in $s_i$ and a point $s_{i+1}$; and let $\angle(s_i,s_{i+1})$ be the angle between the supporting lines of $s_i$ and $s_{i+1}$. 

\begin{lemma}\label{lem:angle}
For all $i\in \{1,\ldots , m-1\}$, we have $\dist(s_i,s_{i+1})\geq \Omega(h(P)/\alpha)$ or $\angle (s_is_{i+1})\geq \Omega(1/\alpha)$. 
\end{lemma}
\begin{proof}
By construction, there is an $\alpha$-fat convex body $a_i\in \across(P)$ such that $a_i\cap \slab(P)\subset \tau_i$. By Lemma~\ref{lem:inscribed}, $a_i\cap \slab(P)$ contains a disk $d$ of radius $\Omega(h(P)/\alpha)$. Note that $d\subset a_i\subset \tau_i$. If $\dist(s_i,s_{i+1})\geq \frac12\, \diam(d)\geq \Omega(h(P)/\alpha)$, then the proof is complete. 

Suppose, for contradiction, that $\dist(s_i,s_{i+1})< \frac12\, \diam(d)$ and $\angle (s_i,s_{i+1})<1/(5\alpha)$. Let $d_{\rm in}$ be a maximum inscribed disk of $a_i$, with radius $\varrho_{\rm in}$ and center $c_{\rm in}$. By maximality, $\diam(d_{\rm in})\geq \diam(d)$. The angle between the two tangents from $p$ to $d_{\rm in}$ is 
at least $\angle (s_i,s_{i+1})$, which implies 
\begin{equation}\label{eq:cp2}
\|c_{\rm in} p\| 
= \frac{\varrho_{\rm in}}{\sin (\frac12\ \angle (s_i,s_{i+1}))}
> \frac{\varrho_{\rm in} }{\sin (1/(10\alpha))}
\geq 10\alpha\, \varrho_{\rm in},
\end{equation}
using the Taylor estimate $\sin(x)\leq x$ for $0\leq x\leq \pi$.

Let $b$ be an arbitrary point in $a_i\cap b_P$. In the remainder of the proof we derive a lower bound for $\|bc_{\rm in} \|$ to reach a contradiction. 
The minimum distance between two line segments is attained between an endpoint of one segment and some point in the other. Assume w.l.o.g.\ that $\dist(s_i,s_{i+1})$ is attained between $q=b_P\cap s_i$ and some point $r\in s_{i+1}$. Then, by assumption, $\|qr\|=\dist(s_i,s_{i+1})< \frac12\, \diam(d)$. As $r\in \slab(P)$, then the slope of $qr$ is nonnegative. 
Dilate the triangle $\Delta(p,q,r)$ from center $p$ to a triangle $\Delta(p,q',r')$ such that $q' r'$ passes though $c_{\rm in}$. Since $\|qr\|<\frac12\, \diam(d_{\rm in})$ and $\|q' r'\| > \diam(d_{\rm in})$, then the scaling factor is more than 2. In particular, we obtain $2\, \|rr'\|> \|pr'\|$. 
Taking the orthogonal projections of segments $c_{\rm in} b$ and $c_{\rm in}p$ to the supporting line of $s_{i+1}$, and using \eqref{eq:cp2}, we obtain 
\[
\|bc_{\rm in}\|
\geq \|rr'\|
> \frac{\|pr'\|}{2} 
= \frac{\|c_{\rm in} p\| \cos \angle (cp,s_{i+1})}{2}   
>  5\alpha\, \varrho_{\rm in} \cos \left(\frac{1}{5\alpha}\right) 
> 4\alpha\, \varrho_{\rm in}. 
\]
Denote by $\varrho_{\rm out}$ the radius of a minimum enclosing disk of $a_i$. Then  $\varrho_{\rm out}\geq \frac12\, \diam(a_i)\geq \frac12\, \|bc_{\rm in}\| \geq 2\alpha\, \varrho_{\rm in}$, which yields $\varrho_{\rm out}/\varrho_{\rm in}\geq 2\alpha$,  contradicting the assumption that $a_i$ is $\alpha$-fat.
\end{proof}

\begin{figure}[htbp]
 \centering
 \includegraphics[width=\textwidth]{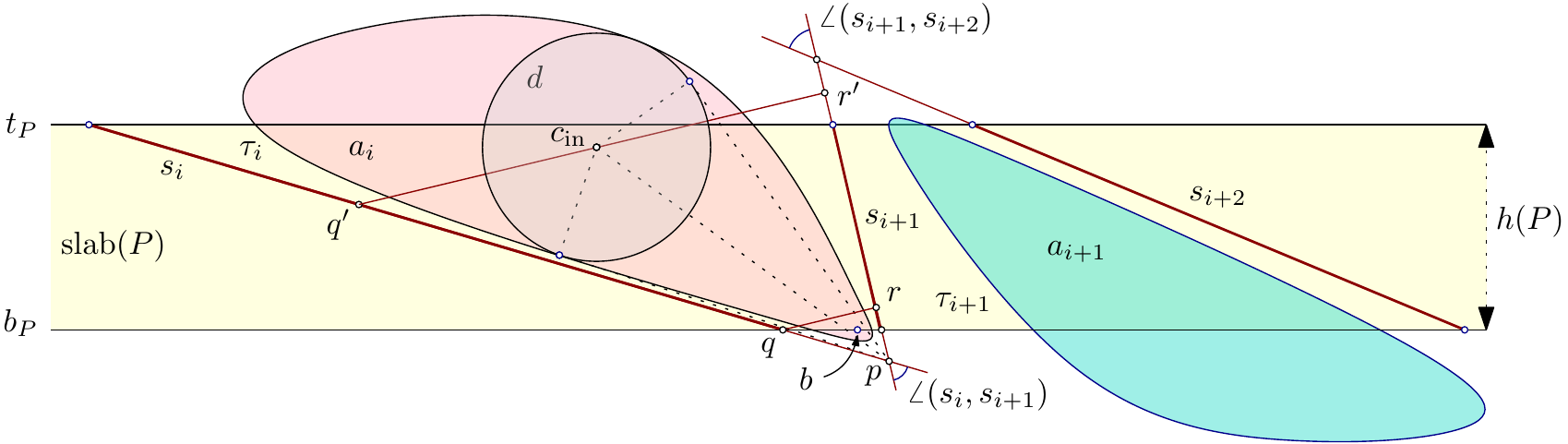}
 \caption{An illustration for the proof of Lemma~\ref{lem:angle} and Lemma~\ref{lem:consecutive}.}
    \label{fig:angles}
\end{figure}

We can distinguish between three \emph{types} of trapezoids: Let $\mathcal{T}_=$ be the set of trapezoids $\tau_i$ with $\dist(s_i,s_{i+1})\geq h(P)/\alpha$. 
Let $\mathcal{T}_<$ (resp., $\mathcal{T}_>$) be the set of trapezoids $t_i$ such that $\dist(s_i,s_{i+1})<h(P)/\alpha$ and the supporting lines of $s_i$ and $s_{i+1}$ meet below (resp., above) $\slab(P)$. 
By Lemma~\ref{lem:angle} there are at most $O(\alpha)$ consecutive trapezoids in $\mathcal{T}_<$ (resp.,  $\mathcal{T}_>$). The next lemma considers the transition between two consecutive trapezoids in $\mathcal{T}_<$ and $\mathcal{T}_>$.

\begin{lemma}\label{lem:consecutive}
If $\tau_i\in \mathcal{T}_<$ and $\tau_{i+1}\in \mathcal{T}_>$ (or, symmetrically, $\tau_i\in \mathcal{T}_>$ and $\tau_{i+1}\in \mathcal{T}_<$), 
then $\dist(s_i,s_{i+2})\geq \Omega(h(P)/\alpha)$. 
\end{lemma}
\begin{proof}
Assume that $\tau_i\in \mathcal{T}_<$ and $\tau_{i+1}\in \mathcal{T}_>$. 
Note that $\|s_{i+1}\|\geq h(P)$. By Lemma~\ref{lem:angle}, we have $\angle(s_i,s_{i+1})\geq \Omega(1/\alpha)$
and $\angle(s_{i+1},s_{i+2})\geq \Omega(1/\alpha)$. Assume w.l.o.g.\ that 
$\angle(s_i,s_{i+1})\leq\angle(s_{i+1},s_{i+2})$. Then $s_i$ and $s_{i+2}$ are separated by a slab bounded by two lines parallel to $s_i$ that pass through the two endpoints of $s_{i+1}$. The distance between these two parallel lines is at least 
$
\|s_{i+1}\| \cdot \sin \angle (s_i,s_{i+1})
\geq h(P)\cdot \sin(\Omega(1/\alpha))
\geq \Omega(h(P)/\alpha),
$
as claimed.
\end{proof}

\begin{lemma}\label{lem:segment}
Every line segment $s\subset \slab(P)$ intersects  $O(\alpha^2\,\mathrm{length}(s)/h(P))$ trapezoids. 
\end{lemma}
\begin{proof}
Suppose that $s$ intersects a sequence of trapezoids $\tau_1,\ldots , \tau_{m}$. Then it crosses the trapezoids $\tau_2,\ldots , \tau_{m-1}$.
If $s$ crosses $\tau_j\in \mathcal{T}_=$, then $\mathrm{length}(s\cap \tau_j)\geq \Omega(h(P)/\alpha)$ by Lemma~\ref{lem:angle}. If it crosses two consecutive trapezoids $\tau_j\in \mathcal{T}_<$ and $\tau_{j+1}\in \mathcal{T}_>$ or vice versa, then $\mathrm{length}(s\cap (\tau_j\cup \tau_{j+1}))\geq \Omega(h(P)/\alpha)$ by Lemma~\ref{lem:consecutive}. We are left with sequences of consecutive trapezoids in either $\mathcal{T}_<$ or $\mathcal{T}_>$. However, by Lemma~\ref{lem:angle}, each such sequence contains at most $O(\alpha)$ trapezoids. Overall, $s$ intersects at most 
\[
2+O\left(\frac{\mathrm{length}(s)}{h(P)/\alpha}\right)\cdot O(\alpha)
\leq O\left(\frac{\alpha^2\,\mathrm{length}(s)}{h(P)}\right)
\] trapezoids.
\end{proof}

\begin{lemma}\label{lem:fat-inside}
Every $\alpha$-fat convex body in $\inside(P)$ intersects $O(\alpha^3)$ bodies in $\across(P)$ that are disjoint w.r.t.\ $\slab(P)$. 
\end{lemma}
\begin{proof}
Let $s\in \inside(P)$; and let $d_{\rm in}$ and $d_{\rm out}$, resp., be a maximal inscribed and a minimum enclosing disk of $s$. Then $s\subset \slab(P)$ and $\alpha$-fatness imply
\[\diam(s)
\leq \diam(d_{\rm out})
\leq \alpha\, \diam(d_{\rm in})
\leq \alpha\, \mathrm{height}(s)
\leq \alpha \, h(P).
\]

Suppose that $s$ intersects a sequence of $m$ convex bodies $a_1,\ldots , a_m\in \across(P)$,  that are disjoint w.r.t.\ $\slab(P)$, and lie in trapezoids $\tau_1,\ldots ,\tau_m$. Let $s_0\subset s$ be a line segment between a point in $a_1\cap s$ and $a_m\cap s$. Then $s$ and $s_0$ intersect the same sequence of trapezoids; and  $\mathrm{length}(s_0)\leq \diam(s)\leq \alpha\, h(P)$. By Lemma~\ref{lem:segment}, 
$s_0$ intersects $O(\alpha^3)$ trapezoids.
\end{proof}

\begin{lemma}\label{lem:fat-boundary}
Every $\alpha$-fat convex body $b\in \bottomset(P)$ (resp., $\topset(P)$) intersects $O(\alpha^3)$  bodies $a \in \across(P)$ that are pairwise disjoint w.r.t.\ $\slab(P)$ and do not intersect $b$ along $b_P$ (resp., $t_P$).
\end{lemma}
\begin{proof}
Let $b\in \bottomset(P)$. Let $l$ and $r$, resp., be the vertical lines passing through the left and right endpoints of $b\cap b_P$. The lines $l$ and $r$ decompose $b$ into three sets: $b_{\rm left}$, $b_{\rm center}$, and $b_{\rm right}$; see Fig.~\ref{fig:sides}.

\begin{figure}[htbp]
 \centering
 \includegraphics[width=\textwidth]{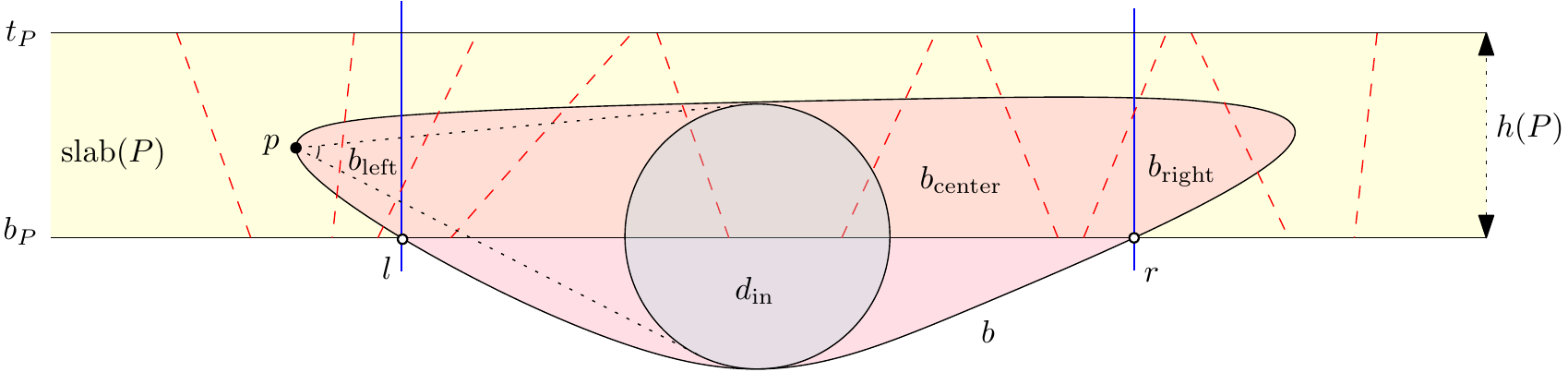}
 \caption{A convex body $b\in \bottomset(P)$ is divided into three parts: 
 $b_\mathrm{left}$, $b_\mathrm{center}$, and $b_\mathrm{right}$. Dashed lines indicate a trapezoid decomposition of $\slab(P)$.}
    \label{fig:sides}
\end{figure}

We claim that $\diam(b_{\rm left})\leq O(\alpha\, h(P))$ and $\diam(b_{\rm right})\leq O(\alpha \, h(P))$. We prove the claim for $b_{\rm left}$; the argument is analogous for $b_{\rm right}$. Consider a maximum inscribed disk $d_{\rm in}$ of $b$. If $\diam(d_{\rm in})\leq 2\, h(P)$, then
$
\diam(b_{\rm left})
\leq \diam(b)
\leq \alpha\, \diam(d_{\rm in}) 
\leq \alpha \cdot 2h(P).
$

Otherwise the radius of $d_{\rm in}$ exceeds $h(P)$. Since $d_{\rm in}\subset b$ and $b\in \bottomset(P)$, then $d_{\rm in}$ lies below $t_P$ and the center of $d_{\rm in}$ is below $b_P$. This, in turn, implies that $d_{\rm in}\cap \slab(P)$ lies between $l$ and $r$; hence it is disjoint from $b_{\rm left}$. Let $p$ be a leftmost point in $b_{\rm left}$. The angle between the two tangents from $p$ to $d_{\rm in}$ is 
$\beta\geq 2\, \sin^{-1}\left(\frac{1}{2\alpha}\right)\geq \Omega(1/\alpha)$. 
By convexity, both tangent lines cross $l$ within $b$, hence within $\slab(P)$;
and so $p$ is at distance at most $(h(P)/2)/\tan(\beta/2)\leq O(\alpha\, h(P))$ from $l$. Now the axis-aligned bounding box of $b_{\rm}$ has height $O(h(P))$ and width $O(\alpha\, h(P))$, which implies that $\diam(b_{\rm left})\leq O(\alpha\, h(P))$, as claimed. 

Suppose that $b$ intersects a sequence of convex bodies $a_1,\ldots , a_m\in \across(P)$ that are pairwise disjoint w.r.t.\ $\slab(P)$ and are also disjoint from $b\cap b_P$. 
Then $a_1,\ldots , a_m$ lie in trapezoids $\tau_1,\ldots ,\tau_m$. Every trapezoid intersects $b_{\rm left}$ or $b_{\rm right}$, or the segments $l\cap \slab(P)$ or $r\cap \slab(P)$. By Lemma~\ref{lem:segment},  they jointly intersect at most $O(\alpha^3)$ trapezoids.
\end{proof}

\begin{lemma}
For every node $P$ of the partition tree, $H(P)$ has at most $O(\alpha^3\cdot |S(P)|)$ edges. 
\end{lemma}
\begin{proof}
Recall that $H(P)=H_\inside(P)\cup H_\bottomset(P)\cup H_\topset(P)$. By construction, $H_{\bottomset}(P)$ and $H_{\topset}(P)$ are 2-hop spanners for the intersection graphs of the intervals $\{s\cap b_P: s\in \bottomset(P)\}$ and $\{s\cap t_P: s\in \topset(P)\}$, respectively. By Theorem~\ref{thm:intervals}, they have at most $2\, |\bottomset(P)|$ and $2\, |\topset(P)|$ edges, respectively. 

By Lemma~\ref{lem:disjoint}, $\mathcal{C}\cup \mathcal{C}'$ comprises four subsets of bodies in $\across(P)$ that are disjoint w.r.t. $\slab(P)$. By Lemma~\ref{lem:fat-inside}, each $s\in \inside(P)$ intersects $O(\alpha^3)$ bodies in $\mathcal{C}\cup \mathcal{C}'$. By Lemma~\ref{lem:fat-boundary}, each $s\in \bottomset(P)\cup \topset(P)$ intersects $O(\alpha^3)$ bodies in $\mathcal{C}\cup \mathcal{C}'$ that $s$ intersects along neither $b_P$ nor $t_P$. Each edge of $H_\inside(P)$ corresponds to such an intersection, hence $H_\inside(P)$ has $O(\alpha^3\, |\inside(P)\cup\bottomset(P)\cup\topset(P)|) \leq O(\alpha^3\, |S(P)|)$ edges. 
Overall, $H(P)$ has $O(\alpha^3\, |S(P)|)$ edges.
\end{proof}
\begin{corollary}\label{cor:number_of_edges}
For every set $S$ of $n$ $\alpha$-fat convex bodies in the plane, $H(S)$ has $O(\alpha^3\, n\log n)$ edges. 
\end{corollary}

The main result of this section is the combination of Corollaries~\ref{cor:correctness_axis} and \ref{cor:number_of_edges}. 

\begin{theorem}\label{thm:fat}
The intersection graph of a set of $n$ $\alpha$-fat convex bodies in the plane admits a 3-hop spanner with $O(\alpha^3\, n\log n)$ edges.
\end{theorem}

\begin{figure}[htbp]
 \centering
 \includegraphics[width=.85\textwidth]{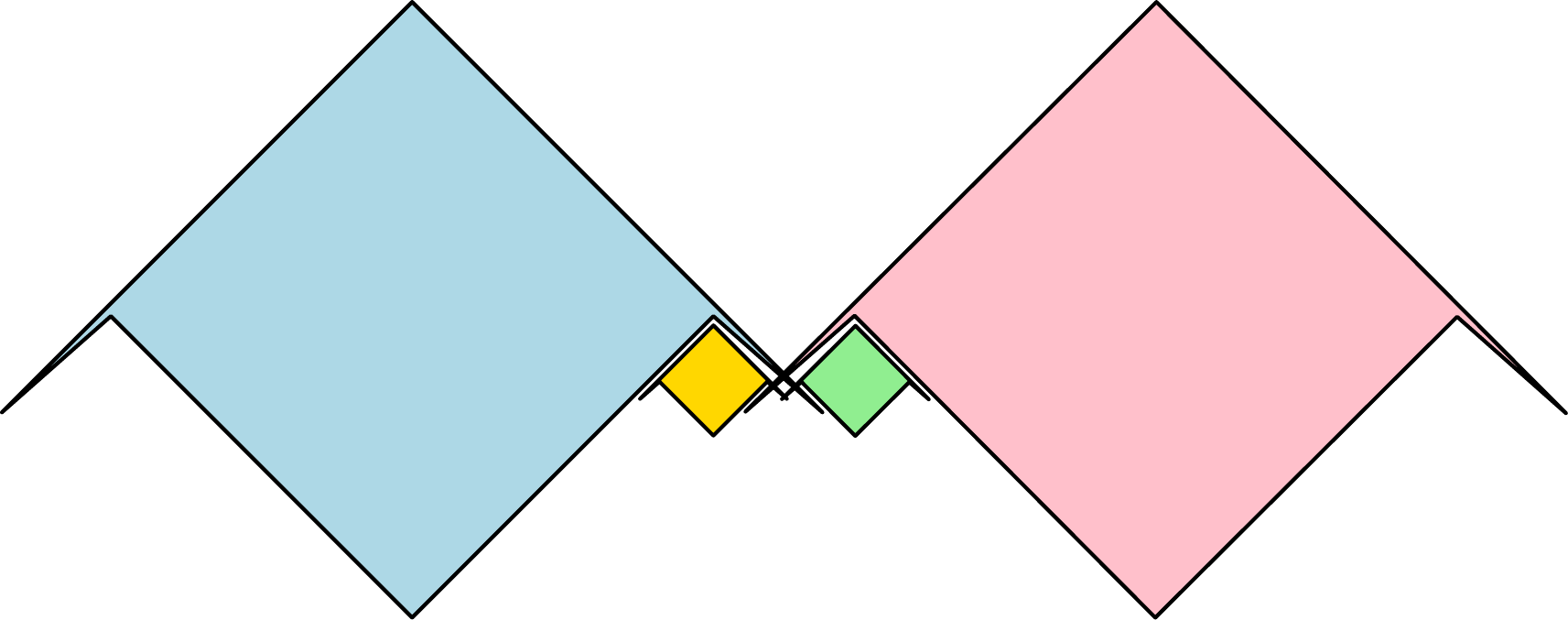}
 \caption{Homothetic copies of a nonconvex fat body $F$, with a biclique intersection graph.}
    \label{fig:nonconvex}
\end{figure}

\begin{remark}
It is unclear whether the hop-stretch factor can be improved to 2, or whether the convexity condition can be dropped. However, the convexity condition cannot be dropped if we wish to find a 2-hop spanner: It is not difficult to construct a fat body $F$ such that the intersection graph of $2n$ homothetic copies of $F$ is $K_{n,n}$; see Fig.~\ref{fig:nonconvex}.
\end{remark}

\section{Three-Hop Spanners for Axis-Aligned Rectangles}
\label{sec:rectangles}

The intersection graph of $n$ arbitrary axis-aligned rectangles may be a complete bipartite graph $K_{\lfloor n/2\rfloor,\lceil n/2\rceil}$, so a 2-hop spanner with subquadratically many edges is impossible. However, we show how to construct a 3-hop spanner with $O(n\log^2 n)$ edges.

Let $\mathcal{R}$ be set of $n$ axis-aligned rectangles in the plane, and let $G$ be their intersection graph. We distinguish between two types of intersections (see Fig.~\ref{fig:boxes}): Two rectangles have a \emph{corner intersection} if one of them contains some corner of the other; and a \emph{crossing intersection} if neither contains any corner of the other. 

\begin{figure}[htbp]
 \centering
 \includegraphics[width=.85\textwidth]{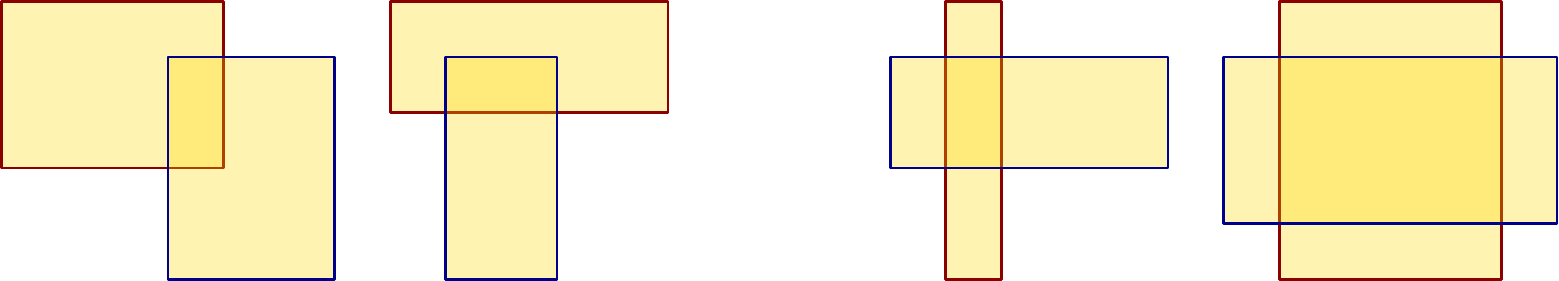}
 \caption{Corner intersections (left)
       and crossing intersections (right)}
    \label{fig:boxes}
\end{figure}

We can partition the intersection graph $G(\mathcal{R})$ into two edge-disjoint subgraphs, $G(\mathcal{R})=G_C(\mathcal{R})\cup G_X(\mathcal{R})$, where two rectangles are adjacent in $G_C(\mathcal{R})$ (resp., $G_X(\mathcal{R})$) iff they have a corner (resp., crossing) intersection. It is clear that if $H_C$ and $H_X$ are $k$-hop spanners for $G_C(\mathcal{R})$ and $G_X(\mathcal{R})$, resp., then $H=H_C\cup H_X$ is a $k$-hop spanner for $G(\mathcal{R})$.

\subsection{Corner Intersections}
\label{ssec:corner}

The divide-and-conquer strategy in Section~\ref{ssec:squares} used the fatness of the rectangles in only one step. We modify the construction to obtain a 2-hop spanner $\widehat{H}$ for corner intersections.

For a set $\mathcal{R}$ of $n$ axis-aligned rectangles, we recursively partition  $\R^2$ into slabs exactly as in Section~\ref{ssec:squares}. For a leaf node $P$, we use the same graph $\widehat{H}(P)=H(P)$. For an internal node $P$, we define the same graphs $H_\bottomset(P),H_\topset(P), H'_\centerset(P)\subset G(\mathcal{R})$, which are 2-hop spanners for $G(\bottomset(P))$, $G(\topset(P))$, and $G(\across(P))$, respectively. 
However, instead of $H_\centerset(P)$, we define $\widehat{H}_\centerset(P)$:
Recall that we subdivided $\slab(P)\cap (\bigcup \across(P))$ into disjoint rectangles, where each rectangle $I_k$ is covered by a rectangle $c_k\in \across(P)$.
We augment $H'_\centerset(P)$ as follows: for every $s \in \inside(P)$, add an edge between $s$ and $c_k$ if a corner of $s$
is in $I_k$; and let $\widehat{H}_{\centerset}(P)$ be the resulting graph. 
Finally let $\widehat{H}(P) = H_{\bottomset}(P) \cup H_{\topset}(P) \cup \widehat{H}_{\centerset}(P)$.

We can now adjust Lemma~\ref{lem:subproblemspanner} for corner intersections. 

\begin{lemma}\label{lem:subproblem-spanner+}
For every subproblem $P$, the following holds:
\begin{enumerate}
    \item For each edge $ab \in G_C(P)$ with $a \in \across(P)$, 
          $\widehat{H}(P)$ contains an $ab$-path of length at most 2.
    \item $\widehat{H}(P)$ contains $O(|S(P)|)$ edges.
\end{enumerate}
\end{lemma}
\begin{proof}
\textbf{(1)} If $b \in \bottomset(P)\cup \topset(P)$, the proof is identical to that of Lemma~\ref{lem:subproblemspanner}. Assume that $b \in \inside(P)$.
For a rectangle $a\in \across$, none of the corners lie in the interior of $\slab(P)$;
while all corners of $b\in \inside(P)$ lie in the interior of $\slab(P)$. 
If $ab$ is an edge in $G_C(P)$, then $a$ and $b$ have a corner intersection,
then $a$ contains some corner $p\in b$. For each corner $p\in b$,  $\widehat{H}_\centerset(P)\subset \widehat{H}(P)$ contains an edge $bc_k$ between $b$ and a rectangle $c_k\in \across(P)$ with $p\in c_k$. It also contains an edge $ac_k$ if $a\cap c_k\neq \emptyset$. Consequently, it contains the path $(a,c_k,b)$, as required.

\smallskip\noindent \textbf{(2)} Recall that  $\widehat{H}(P) = H_{\bottomset}(P) \cup H_{\topset}(P) \cup \widehat{H}_{\centerset}(P)$. By Corollary~\ref{cor_2d_line}, $H_{\bottomset}(P)$, $H_{\topset}(P)$, and $H_{\inside}'(P)$ each have at most $2\,|S(P)|$ edges. The graph $H_{\inside}$ was constructed by augmenting $H_{\inside}'(P)$ with one edge for each of the 4 corners of every rectangle in $\inside(P)$, and so $H_{\inside}$ has at most $4\,|\inside(P)|$ more edges than $H_{\inside}'(P)$. In total, $\widehat{H}(P)$ has $O(|S(P)|)$ edges.
\end{proof}

Lemmata~\ref{lem:square_subproblem_counts}--\ref{lem:correctness_axis} carry over verbatim (they do not use fatness); and yield the following.

\begin{lemma}\label{lem:corner}
Let $\mathcal{R}$ be a set of $n$ axis-aligned rectangles in the plane. There is a graph $\widehat{H}\subset G(\mathcal{R})$ with $O(n\log n)$ edges such that if rectangles $a,b\in \mathcal{R}$ have a corner intersection, then $\widehat{H}$ contains a $ab$-path of length at most 2.
\end{lemma}
Note, however, that $\widehat{H}$ need not be a 2-hop spanner \emph{for} $G_C(\mathcal{R})$, as some of the edges in $\widehat{H}$ may be in $G_X(\mathcal{R})$. In any case, $\widehat{H}$ is a subgraph of $G(\mathcal{R})$.

\subsection{Crossing Intersections}
\label{ssec:crossing}

An axis-aligned rectangle $r$ is the Cartesian product of two closed intervals, $r=r_x\times r_y$. We say that a rectangle $r'$ \emph{traverses $r$ horizontally} if $r_x\subset r'_x$ and $r'_y\cap r_y\neq \emptyset$; and \emph{vertically} if $r'_x\cap r_x\neq \emptyset$ and $r_y\subset r'_y$. If $\mathcal{A}(r)$ and $\mathcal{B}(r)$ are sets of rectangles that traverse $r$ horizontally and vertically, resp., then the intersections between  $\mathcal{A}$ and $\mathcal{B}$ form a biclique in $G(\mathcal{R})$. 

\begin{figure}[htbp]
 \centering
 \includegraphics[width=.85\textwidth]{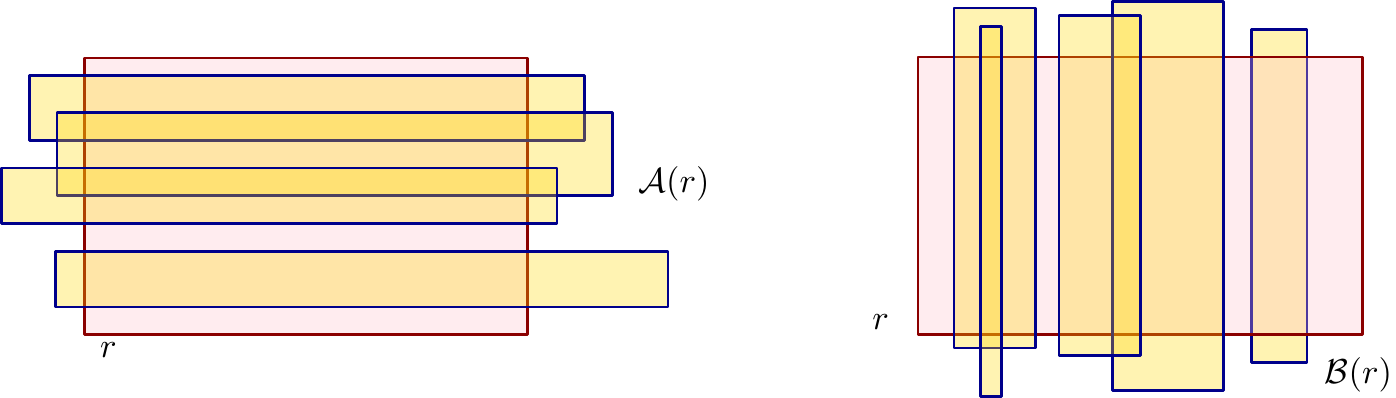}
 \caption{Rectangles that traverse $r$ horizontally (left); and vertically (right)}
    \label{fig:crossings}
\end{figure}

Recall that bicliques admit 3-hop spanners with linearly many edges.

\begin{observation}\label{obs:bistar}
For all $s,t\in \N$, $K_{s,t}$ has a 3-hop spanner with $s+t-1$ edges.  
\end{observation}
\begin{proof}
Let $\mathcal{A}$ and $\mathcal{B}$ be the two partite sets of $K_{s,t}$. For two arbitrary vertices $a_0\in \mathcal{A}$ and $b_0\in \mathcal{B}$, let $H$ be the \emph{bistar} comprising two maximal stars centered at $a_0$ and $b_0$, respectively.
Then $H$ has $s+t-1$ edges, and for any $a\in \mathcal{A}$ and $b\in \mathcal{B}$, it contains a 3-hop path $(a,b_0,a_0,b)$.
\end{proof}

\begin{lemma}\label{lem:crossing}
For a set $\mathcal{R}$ of $n$ axis-aligned rectangles in the plane, the graph $G_X(\mathcal{R})$ of crossing intersections can be covered by edge bicliques in $G(\mathcal{R})$ 
of total weight $W=O(n\log^2 n)$, with weight function $w(K_{s,t})=s+t$.
\end{lemma}
\begin{proof}
We may assume w.l.o.g.\ that $n=2^k$ for some $k\in \N$. Note that the intersections between rectangles depend only on the order of the $x$- and $y$-coordinates of their vertices. By applying an order-preserving homeomorphism on the $x$-coordinates (resp., $y$-coordinates), we may assume w.l.o.g.\ that the vertices of all rectangles in $\mathcal{R}$ have integer coordinates in $\{0,1,\ldots , 2n-1\}$. Let $\overline{i,j}:=\{i,i+1,\ldots , j\}$ denote a set consecutive integers, called a \emph{streak} of length $j-i+1$. 

We construct an edge-cover of the graph $G_X(\mathcal{R})$ of crossing intersections with bicliques in $G(\mathcal{R})$ such that each rectangle in $\mathcal{R}$ is part of $O(k^2)=O(\log^2 n)$ bicliques. Partition the streak $\overline{0,2n-1}$ recursively into streaks of length $2^i$ for $i=k+1,k,k-1,\ldots ,0$. Let $T$ denote the recursion tree (i.e., the root of $T$ corresponds to the streak $\overline{0,2n-1}$, its children to $\overline{0,n-1}$ and $\overline{n,2n-1}$, and the leaves to singletons); see Fig.~\ref{fig:hierarchy}. For a nonroot streak $S\in T$, let $p(S)$ be the parent streak of $S$. 
Note that the Cartesian product of two streaks of length $2^i$ and $2^j$, resp., is a $2^i\times 2^j$ section of the integer lattice (\emph{grid}, for short). For all $i,j\in \{1,\ldots , k+1\}$, let $\mathcal{L}_{i,j}$ be the set of $2^i\times 2^j$ grids formed by the Cartesian product of two streaks in $T$ of length $2^i$ and $2^j$, respectively. 

\begin{figure}[htbp]
 \centering
 \includegraphics[width=.5\textwidth]{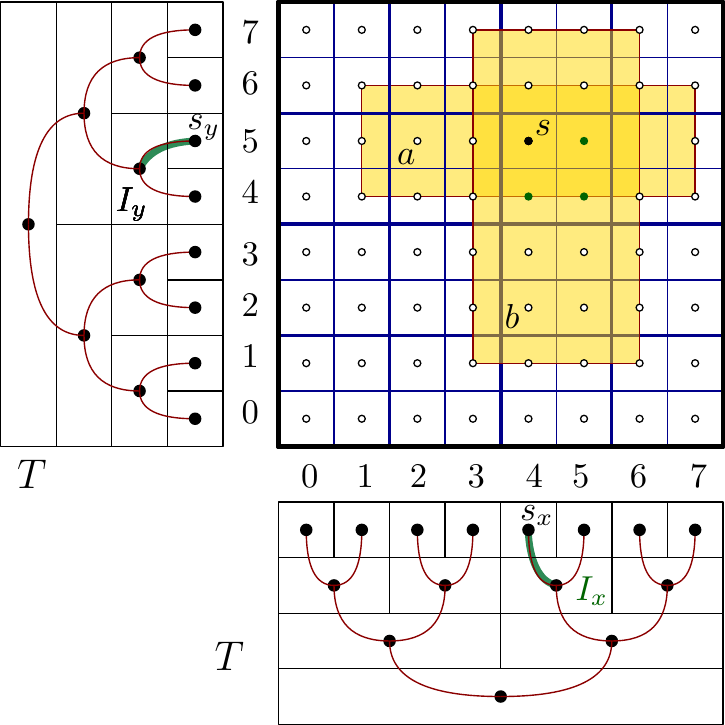}
 \caption{The recursion tree $T$ on the $x$- and $y$-coordinates for $n=4$. Rectangles $a$ and $b$, with integer coordinates, have a crossing intersection. An integer point $s\in I_x\times I_y\subset a\cap b$.}
    \label{fig:hierarchy}
\end{figure}

For each $2^i\times 2^j$ grid $\ell\in \mathcal{L}_{i,j}$, let $\mathcal{A}(\ell)$ be the set of rectangles $r\in \mathcal{R}$ such that $\ell_x\subset r_x$ but $p(\ell_x)\not\subset r_x$; and $r_y\cap \ell_y\neq \emptyset$.
Similarly, let $\mathcal{B}(\ell)$ be the set of rectangles $r\in \mathcal{R}$ 
such that $\ell_y\subset r_y$ but $p(\ell_y)\not\subset r_y$; and $r_x\cap\ell_x\neq\emptyset$.
Then every rectangle in $\mathcal{A}(\ell)$ (resp., $\mathcal{B}(\ell)$) traverse the convex hull of the grid $\ell$ horizontally (resp., vertically), and so the bipartite intersection graph between $\mathcal{A}(\ell)$ and $\mathcal{B}(\ell)$ is a biclique $G(\ell)\subset G(\mathcal{R}$).  Note, however, that the intersections between rectangles in $\mathcal{A}(\ell)$ and $\mathcal{B}(\ell)$ need not be crosssing intersection, and so $G(\ell)$ is not necessarily a subgraph of $G_X(\mathcal{R})$.

\smallskip\noindent\textbf{Correctness.} We show that the bicliques $G(\ell)$, for all $\ell\in \mathcal{L}_{i,j}$ and $i,j\in \{1,\ldots ,k+1\}$, cover all edges of  $G_X(\mathcal{R})$. Assume that $a,b\in \mathcal{R}$ have a crossing intersection.
We may assume w.l.o.g.\ that $b_x\subset a_x$ and $a_y\subset b_y$.
Let $s\in a\cap b$ be an intersection point with integer coordinates. 
Then $s_x$ and $s_y$ correspond to two leaves in the recursion tree $T$. 
Following a leaf-to-root path in $T$, we find the maximal streak $I_x$ 
such that $s_x\in I_x\subset a_x$. 
Similarly, let $I_y\in T$ be the maximal streak in $T$ such that $s_y\in I_y\subset b_y$. 
Then for the grid $\ell =I_x\times I_y$, we have $a\in \mathcal{A}(\ell)$ 
as $I_x\subset a_x$ but $a_x\not\subset p(I_x)$; and $b\in \mathcal{B}(\ell)$ as $I_y\subset b_y$ but $b_y\not\subset p(I_y)$. Now $ab$ is an edge in the biclique $G(\ell)$. 

\smallskip\noindent\textbf{Weight analysis.} 
For every $r\in \mathcal{R}$, $r_x$ and $r_y$ are each partitioned into at most $2k\leq O(\log n)$ maximal intervals in $T$. Consequently, each $r\in \mathcal{R}$ participates in $O(\log^2 n)$ bicliques. The sum of orders of all bicliques is $O(n \log^2 n)$. 
\end{proof}

\begin{corollary}\label{cor:crossing}
For a set $\mathcal{R}$ of $n$ axis-aligned rectangles in the plane, the graph $G_X(\mathcal{R})$ admits a 3-hop spanner with $O(n\log^2 n)$ edges.
\end{corollary}
\begin{proof}
Let $\mathcal{K}$ be an edge biclique cover of $G_X(\mathcal{R})$ of total weight $O(n\log^2 n)$. For every biclique $K_{s,t}\in \mathcal{K}$, Observation~\ref{obs:bistar} yields a 3-hop spanner $H(K_{s,t})$ with $s+t-1$ edges. Consequently, $H=\bigcup_{K\in \mathcal{K}} H(K)$ is a 3-hop spanner for $G_X(\mathcal{R})$ with at most $\sum_{K\in \mathcal{K}} |E(K)|\leq O(n\log^2 n)$ edges.
\end{proof}

The combination of Lemma~\ref{lem:corner} and Corollary~\ref{cor:crossing} immediately implies the following. 

\begin{theorem}\label{thm:rectangles}
The intersection graph of a set of $n$ axis-aligned rectangles in the plane admits a 3-hop spanner with $O(n\log^2 n)$ edges.
\end{theorem}

\section{Lower Bound Constructions}
\label{sec:lb}

In this section, we define a class of $n$-vertex graphs for which any 2-hop spanner has at least $\Omega(n \log n / \log \log n)$ edges, then show that these graphs can be realized as the intersection graphs of $n$ homothets of any convex body in the plane.

\paragraph{Intuition for the Graph Construction.}
Before a formal description, we give a rough sketch of the construction. We construct a graph $F$ on $n$ vertices where every vertex is associated with an integer in $\{0, \ldots, m\}$, where $m\approx n/\log n$. We show that every 2-hop spanner for $F$ must contain a near-linear number of edges between the vertices labeled $\{0, \ldots, \frac{m}{2}-1\}$ with the vertices labelled $\{\frac{m}{2}, \ldots, m\}$.
Similarly, we show that every 2-hop spanner must contain a large number of edges between the vertices $\{0, \ldots, \frac{m}{4}-1\}$ with the vertices $\{\frac{m}{4}, \ldots, \frac{m}{2}-1\}$, or to connect the vertices $\{\frac{m}{2}, \ldots, \frac{3m}{4}-1\}$ with the vertices $\{\frac{3m}{4}, \ldots, m\}$.
In general, at level $i$, we equipartition $\{0, \ldots, m\}$ into $2^i$ parts, and show that every 2-hop spanner for $F$ must contain $\Omega(i n/\log n)$ edges between consecutive parts.
If the sets of edges required on each level were pairwise disjoint, then summation over all $\log m=\Theta(\log n)$ levels would yield an $\Omega(n \log n)$ lower bound for the spanner size. Unfortunately, the edges required on different levels need not be disjoint. However, we can show that the edges required on any two levels \emph{which are separated by at least $\log \log n$ levels} are disjoint. This gives us an $\Omega(n \log n / \log \log n)$ lower bound on the spanner size.

\paragraph{Construction of $F(h)$.}
For every $h\in \N$, we construct a graph $F(h)$, which contains $2^h (h+1)$ vertices. The vertex set is $V=\{0, \hdots, 2^h - 1\}\times \{0, \hdots, h\}$. For each vertex $v=(x,i)$, we call $i$ the \emph{level} of $v$.
For each level $i \in \{0, \hdots, h\}$, partition the vertices with level less than or equal to $i$ into $2^i$ groups of $2^{h-i} (i+1)$ consecutive vertices based on their $x$-coordinates. In particular, for every level $i \in \{0, \hdots, h\}$, let 
$\{0,\ldots ,2^h-1\}=\bigcup_{k=1}^{2^i-1} X_{k,i}$, where $X_{k, i} = \{2^{h-i} k, 2^{h-i} k + 1, \hdots, 2^{h-i}(k+1) - 1\}$. This defines groups 
$V_{k,i} = X_{k,i}\times \{0,\ldots , i\}$ for $k \in \{0,\ldots , 2^i-1\}$.
Notice that $(x, \ell)\in V_{k, i}$ for $k=\lfloor x/2^i \rfloor$ and $i \geq \ell$.
Finally, add edges to the graph $F(h)$ such that every group $V_{k, i}$ is a clique; see Fig.~\ref{fig:cliques}.

\begin{figure}[htbp]
 \centering
 \includegraphics[width=.5\textwidth]{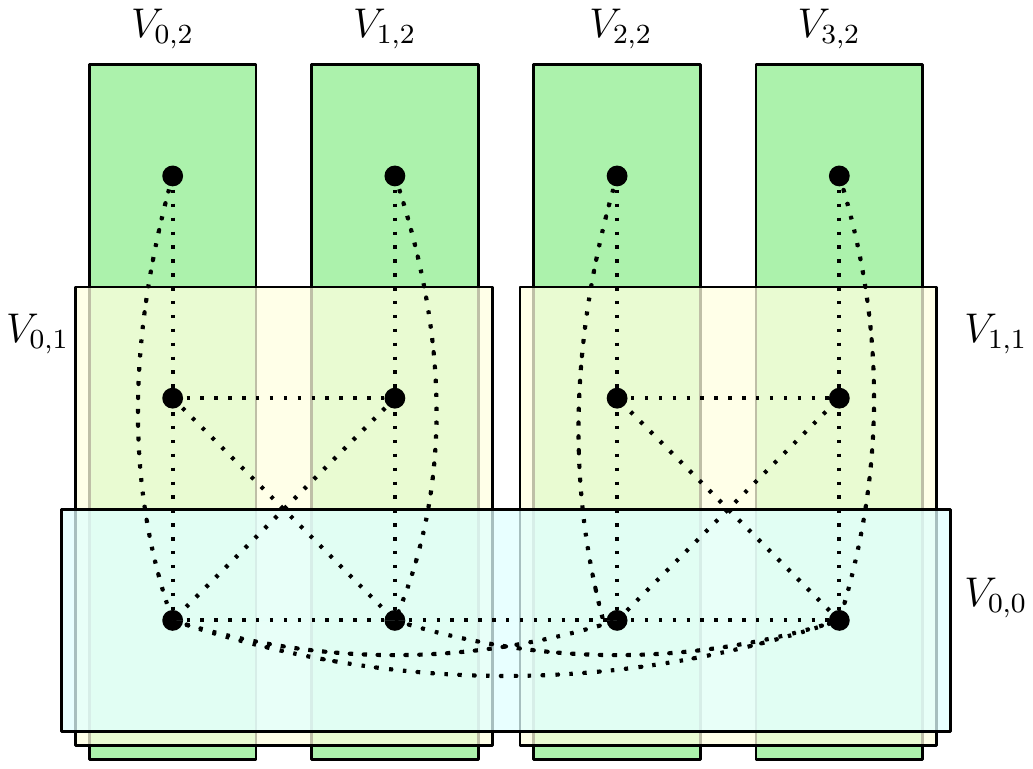}
 \caption{Graph $F(2)$ grouped by cliques $V_{k, i}$}
    \label{fig:cliques}
\end{figure}
We show that any 2-hop spanner for $F(h)$ with $n=2^h (h+1)$ vertices has $\Omega(2^h h^2 / \log h) = \Omega(n \log n / \log \log n)$ edges. We do this by finding a set of cliques $V_{k,i}$ in $F(h)$ such that each $V_{k,i}$ forces at least $\Omega(|V_{k,i}|)$ distinct edges to appear in the spanner. The proof is inspired by the following lemma:


\begin{lemma}\label{lem:bipartite_spanner}
Suppose that the vertex set of the complete graph $K_{2n}$ is partitioned into two sets $A$ and $B$ each of size $n$, and call edges between $A$ and $B$ are bichromatic. Then every 2-hop spanner of $K_{2n}$ contains at least $n$ bichromatic edges.
\end{lemma}
\begin{proof}
Let $S$ be a 2-hop spanner for $K_{2n}$. If  every vertex in $A$ is incident to a bichromatic edge in $S$, then clearly $S$ contains at least $|A| = n$ bichromatic edges. Otherwise, there is some $a \in A$ that has no neighbor in $B$ in $S$. For every $b \in B$, $S$ contains a 2-hop path between $a$ and $b$, that is, a path $(a, a_b, b)$ for some $a_b \in A$. The edges $a_b b$ are bichromatic and distinct for all $b \in B$, so $S$ contains at least $|B| = n$ bichromatic edges.
\end{proof}

Notice that every  $X_{k, i}$, for $i<h$, can be written as $X_{2k, i + 1} \cup X_{2k+1, i + 1}$. Accordingly, we can partition $V_{k,i}$ into two sets of equal size, $V_{k,i} = A \cup B$, where
\[
A = \Big(X_{2k, i + 1}\times \{0, \ldots,  i\}\Big) \quad \text{and} \quad B = \Big(X_{2k+1, i + 1}\times \{0, \ldots, i\}\Big). 
\]
Define a \emph{$V_{k, i}$-bichromatic} edge to be an edge that crosses between $A$ and $B$. 

\begin{lemma}
\label{lem:bichromatic_disjoint}
The set of all $V_{k, i}$-bichromatic edges and the set of all $V_{k', i'}$-bichromatic edges are disjoint unless $k = k'$ and $i' = i$.
\end{lemma}
\begin{proof}
If $i = i'$, then the claim follows from the fact that $V_{k, i}$ and $V_{k', i}$ are disjoint. Otherwise, assume  w.l.o.g.\ that $i < i'$. Notice that either $X_{k', i'}$ is contained within $X_{2k, i + 1}$ or $X_{2k+1, i + 1}$, or it is disjoint from both. The $V_{k, i}$-bichromatic edges cross from $X_{2k, i + 1}$ to $X_{2k+1, i + 1}$ while $V_{k', i'}$-bichromatic edges stay within $X_{k', i'}$, so the edge sets are disjoint.
\end{proof}

Notice that for every clique $V_{i,k}$, Lemma~\ref{lem:bipartite_spanner} implies that a 2-hop spanner on the complete graph induced by $V_{i,k}$ contains at least $|A|=|B|=|V_{i,k}|/2$ edges that are $V_{i,k}$-bichromatic.
Naively, we might hope that for every clique $V_{i,k}$, a 2-hop spanner \emph{on $F(h)$} contains at least $|V_{i,k}|/2$ edges that are $V_{i, k}$-bichromatic. Then, using Lemma~\ref{lem:bichromatic_disjoint}, we could conclude that the 2-hop spanner contains many edges. 
However, this hope is not true:
Lemma~\ref{lem:bipartite_spanner} does not apply when we allow the spanner to be a subgraph of $F(h)$ rather than requiring the spanner to be a subgraph of the complete graph induced by $V_{i,k}$. This is because a 2-hop spanner on $F(h)$ may connect two vertices in $V_{i,k}$ via a vertex outside of $V_{i,k}$; see Figure~\ref{fig:bichromatic_problem}.
To deal with this, we introduce a more technical lemma.

\begin{figure}[htbp]
    \centering
    \includegraphics[width=.5\textwidth]{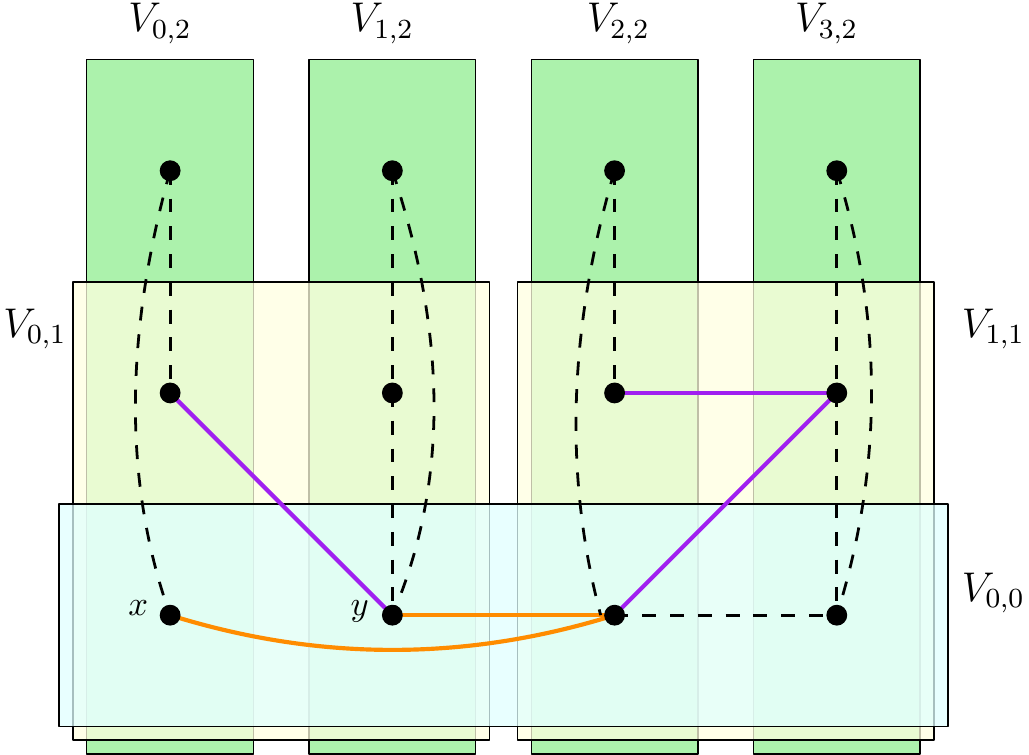}
    \caption{A 2-hop spanner for $F(2)$. Orange edges are $V_{0,0}$-bichromatic. Purple edges are $V_{0,1}$-bichromatic or $V_{1,1}$-bichromatic. There is only 1 edge that is $V_{0,1}$-bichromatic (rather than the 2 bichromatic edges one might expect from Lemma~\ref{lem:bipartite_spanner}), because vertices $x$ and $y$ in $V_{0,1}$ are connected via a 2-hop path using a vertex outside of $V_{0,1}$.}
    \label{fig:bichromatic_problem}
\end{figure}

\begin{lemma}
\label{lem:consecutive_bichromatic}
Let $S$ be a 2-hop spanner of $F(h)$. For all levels $i \in \{0, \ldots, h - \lceil\log h\rceil\}$ and all $k$, the spanner $S$ contains at least $|V_{k, i}|/8$ edges that are $V_{k', i'}$-bichromatic for some $i' \in \{i, \ldots, i + \lceil \log h \rceil - 1\}$. We say that these edges are \emph{forced} by $V_{k, i}$.
\end{lemma}

\begin{proof}
Let $V_{k, i} = A \cup B$ as above. Let $\widehat{A} = X_{2k, i+1} \times \{i\}$ denote the set of vertices in $A$ with the highest level; see Fig.~\ref{fig:bichromatic}. Notice that $|V_{k, i}| = 2^{h-i}(i+1)$ and $|\widehat{A}| = 2^{h-i-1}$.

For each vertex $\hat{a} \in \widehat{A}$, the spanner $S$ contains a path of length at most 2 between $\hat{a}$ and each vertex $b \in B$. We claim that all such paths are in the clique induced by $V_{k, i}$: Indeed, consider a path $(\hat{a}, v, b)$. The fact that $v$ is adjacent to $\hat{a}$ implies that the $x$-coordinate of $v$ is in $X_{k, i}$; this, together with the fact that $v$ is adjacent to $b$, implies that $v$ has level at most $i$.

If $S$ contains at least $|V_{k, i}|/8 = |B|/4$ edges that are $V_{k, i}$-bichromatic, then the proof is complete. Otherwise, let $\hat{a}$ be a vertex in $\widehat{A}$ incident to the fewest $V_{k, i}$-bichromatic edges. In this case, $\hat{a}$ is incident to fewer than $|V_{k, i}|/(8 |\widehat{A}|) = (i+1)/4$ edges that are $V_{k, i}$-bichromatic. Consider the paths of length 1 or 2 from $\hat{a}$ to all vertices $b \in B$ in the spanner $S$. If $b$ is not incident to any $V_{k, i}$-bichromatic edges, which holds for at least $3|B|/4$ vertices in $B$, then $S$ contains a path $(\hat{a}, v, b)$ for some $v \in B$ adjacent to $\hat{a}$. Since $\hat{a}$ has fewer than $(i+1)/4$ neighbors in $B$, then at least $3|B|/4$ vertices in $B$ are covered by at most $(i+1)/4$ stars centered in $B$. In total, the stars contain at least $3|B|/4 - (i+1)/4$ edges.

We claim that at least half of the edges in these stars are $V_{k', i'}$-bichromatic for some $i' \in \{i, \ldots, i + \lceil \log h \rceil - 1\}$. To show this, let $j = i + \lceil \log h \rceil$. The cliques $V_{k', j}$ partition $B$ into $2^{\lceil \log h \rceil - 1}$ subsets, each of size $|B|/2^{\lceil \log h \rceil - 1}$. If one of the stars is centered at $b\in B \cap V_{k', j}$, then at most $|B|/2^{\lceil \log h \rceil - 1} - 1$ of its edges are within $B \cap V_{k', j}$. Over $(i + 1)/4$ stars, the number of edges within the same clique is at most
\[
\left(\frac{|B|}{2^{\lceil \log h \rceil - 1} } - 1\right) \frac{(i + 1)}{4} 
= \frac{|B|(i+1)}{2^{\lceil \log h \rceil + 1} }- \frac{i+1}{4} 
< \frac{|B|}{2} - \frac{i+1}{4}.
\]
Consequently, the remaining $|B|/4$ or more edges of the stars are $V_{k', i'}$-bichromatic for some $i' \in \{i, \ldots, i + \lceil \log h \rceil - 1\}$.

\begin{figure}[htbp]
 \centering
 \includegraphics[width=.6\textwidth]{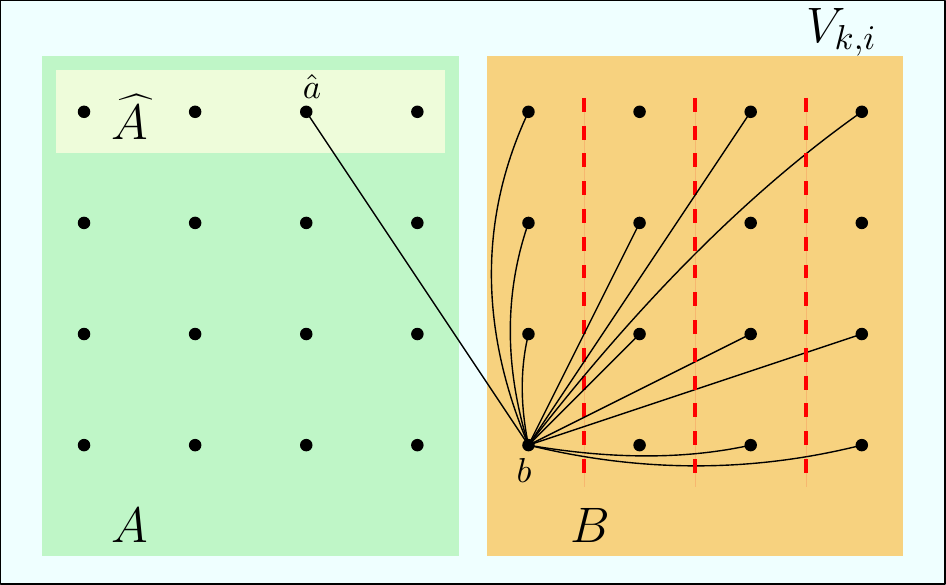}
 \caption{A clique $V_{k, i}$ and the associated sets $A$, $B$, and $\widehat{A}$. Dashed red lines indicate the cliques at level $j = i + \lceil \log h \rceil$. The vertex $\hat{a} \in \widehat{A}$ has at most $(i+1)/4$ bichromatic edges (here, 1 bichromatic edge). The star centered at $b$ is adjacent to at least $3|B|/4$ vertices in $B$, and it has at least $|B|/4$ edges that do not stay within the same clique at level $j$.}
    \label{fig:bichromatic}
\end{figure}

\end{proof}

\begin{lemma}
\label{lem:abstract_lb}
Let $S$ be a 2-hop spanner of $F(h)$. Then $S$ has $n = 2^h(h+1)$ vertices and at least $\Omega(n \log n / \log \log n)$ edges.
\end{lemma}
\begin{proof}

On every $\lceil \log h\rceil$-th level, consider the set $C_i$ of all cliques $V_{k, i}$ at level $i$. By Lemma~\ref{lem:consecutive_bichromatic}, each $V_{k, i}$ forces at least $|V_{k, i}|/8$ edges to appear in the spanner $S$, where each edge is $V_{k', i'}$-bichromatic for some $k'$ and $i'$ with $i' \in \{i, \ldots, i + \lceil \log h \rceil - 1\}$. For a fixed $C_i$, the cliques $V_{k,i}$ in $C_i$ are pairwise disjoint, so the sets of edges that they force in $S$ are pairwise disjoint. Thus, $C_i$ forces $\sum_k |V_{k, i}|/8$ edges in the spanner $S$. At level $i$, the graph $F(h)$ contains $2^i$ cliques $V_{k,i}$, each of size $2^{h-1}(i+1)$. Thus, $C_i$ forces $2^{h-3}(i+1)$ edges in the spanner $S$. Each of these edges is $V_{k', i'}$-bichromatic for some level $i' \in \{i, \ldots, i + \lceil \log h \rceil - 1\}$.

By Lemma~\ref{lem:bichromatic_disjoint}, the edges forced by $C_i$ and $C_{i'}$ are disjoint so long as $|i -i' | \ge \lceil \log h \rceil$. 
%
%
%
%
%
Summation over every $\lceil \log h \rceil$-th level $i$ shows that $S$ contains at least $\Omega \left(\sum_{i = 0}^{h/\log h}  2^{h-3} (i \log h +1)\right) = \Omega(2^{h} h^2 /\log h)= \Omega(n \log n / \log \log n)$ edges. 
\end{proof}

\paragraph{Geometric Realization of $F(h)$.} We realize $F(h)$ as the intersection graph of a set $S(h)$ of homothets of any convex body for all  $h\in \N$. The construction is recursive. To construct $S(h+1)$, we form two copies of $S(h)$ to realize vertices in the first $h$ levels, then add homothets to realize the vertices in level $h+1$.

\begin{lemma}\label{lem:realized_lb}
For every convex body $C \subset \mathbb{R}^2$ and every $h \in \mathbb{N}$, the $n$-vertex graph $F(h)$ can be realized as the intersection graph of a set $S(h)$ of $n$ homothets of $C$.
\end{lemma}
\begin{proof}
Let $C$ be a convex body (i.e., a compact convex set with nonempty interior) in the plane. Let $o\in \partial C$ be an extremal point of $C$. Then there exists a (tangent) line $L$ such that $C\cap L=\{o\}$. Assume w.l.o.g.\ that $o$ is the origin, $L$ is the $x$-axis, and $C$ lies in the upper halfplane. We construct $S(h)$ recursively from $S(h-1)$. Let $s(a, i) \in S(h)$ denote the homothet that represents the vertex $(a, i) \in F(h)$. We maintain two invariants: (I1) for every $a \in \{0, \ldots, 2^h - 1\}$, there is some point $p_a$ on the $x$-axis such that every $s(a,i) \in S(h)$ is tangent to the $x$-axis and intersects the $x$-axis exactly at $p_a$; and (I2) whenever $s_1, s_2 \in S(h)$ intersect, $s_1 \cap s_2$ has nonempty interior.

\smallskip\noindent\textbf{Construction.} 
$F(0)$ has a single vertex $(0,0)$ and no edges, so it can be represented as the single convex body $C$ with the extremal point $o$ on the $x$-axis.

We now construct $S(h)$ from $S(h-1)$; see Fig.~\ref{fig:realization}. By invariant (I2), there is some $\varepsilon > 0$ such that for every $s \in S(h-1)$, translating $s$ by $\varepsilon$ in any direction does not change the intersection graph. Duplicate $S(h-1)$ to form the sets $S_1(h-1)$ and $S_2(h-1)$, and translate every homothet in $S_2(h-1)$ by $\varepsilon$ in the positive $x$ direction. Let $S'(h) = S_1(h-1) \cup S_2(h-1)$. Notice that for every clique $V_{k, i}$ in $F(h-1)$, there is a corresponding clique in the intersection graph of $S'(h)$ that contains both the vertices in the clique $V_{k, i}$ realized by $S_1(h-1)$ \emph{and} the vertices in the clique $V_{k, i}$ realized by $S_2(h-1)$.

The $x$-axis is still tangent to all $s \in S'(h)$, and there are $2^{h}$ distinct points on the $x$-axis that intersect some $s \in S'(h)$. Each point $p_a$ has a neighborhood that intersects only the homothets in $S'(h)$ that contain $p_a$, since every convex body that does not contain $p_a$ has a positive distance from $p_a$ by compactness. For each $p_a$, add a homothetic copy $C_a$ of $C$ completely contained within that neighborhood, tangent to the $x$-axis and containing $p_a$. Let $S(h)$ be the union of $S'(h)$ and these $C_a$.

\begin{figure}[htbp]
 \centering
 \includegraphics[width=.7\textwidth]{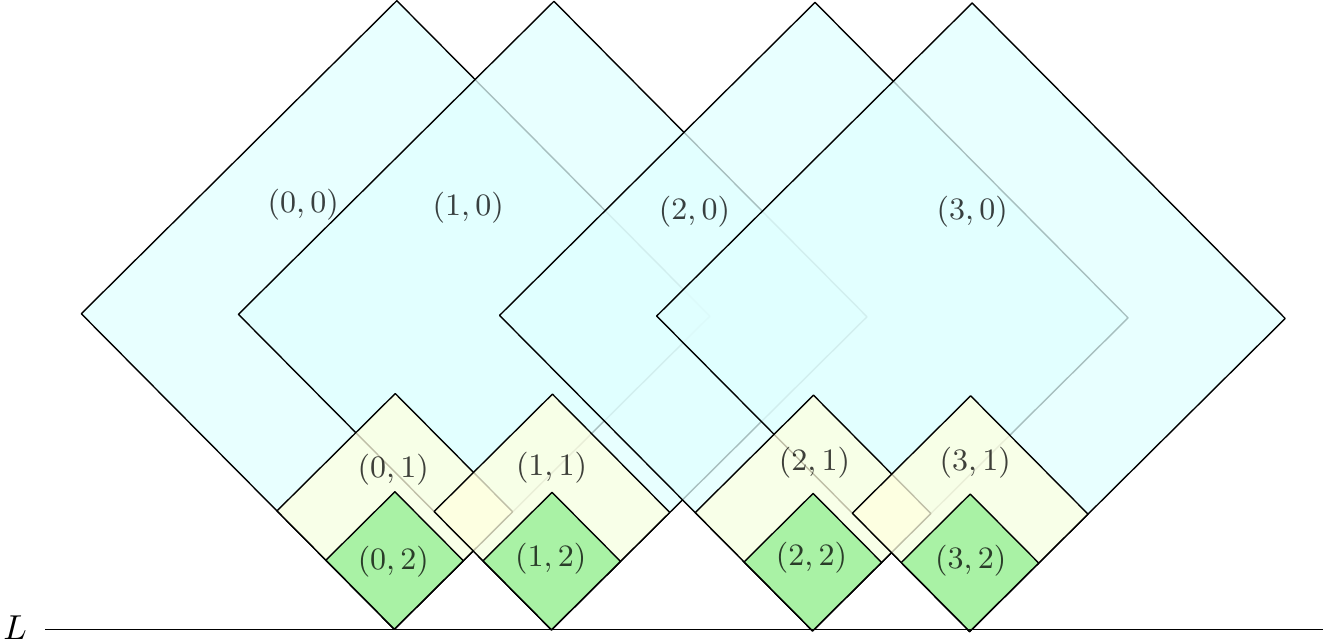}
 \caption{Realization of $F(2)$ with homothets of squares, all tangent to $L$. Each homothet is labeled with the vertex of $F(2)$ that it represents.}
    \label{fig:realization}
\end{figure}

\smallskip\noindent\textbf{Correctness.} 
For $0\leq i < h$, let the homothet $s_{1}(a, i) \in S_1(h-1)$ represent $(2a, i)$ in $F(h)$, and let the homothet $s_{2}(a, i) \in S_2(h-1)$ represent $(2a + 1, i)$ in $F(h)$. Let the homothets $C_a$ represent $(a, h) \in F(h)$.

This correspondence implies that, for all $0\leq i < h$, vertices in the clique $V_{k, i}$ in $F(h)$ have been realized by homothets corresponding to a clique $V_{\lfloor k/2 \rfloor, i}$ in $S_1(h-1)$ or $S_2(h-1)$. 
By construction of $S(h)$, any two such homothets intersect. Similar reasoning applies in the opposite direction: any intersection between two homothets in $S'$ corresponds to an edge in some clique in $F(h)$. When $i = h$, notice that every $C_a$ intersects exactly the homothets in $S'(h)$ that intersect $p_a$, which by assumption were the homothets representing points with the same $x$-coordinate. Thus, any clique $V_{k, h}$ in $F(h)$ is represented in $S(h)$, and there are no edges involving $C_a$ that do not correspond to such a clique in $F(h)$.
\end{proof}

The previous two lemmata imply the following theorem.
\begin{theorem}\label{thm:lb}
For every convex body $C\subset \R^2$, there exists a set $S$ of $n$ homothets of $C$ such that every 2-hop spanner for the intersection graph of $S$ has $\Omega(n\log n / \log \log n)$ edges.
\end{theorem}

We note that it is not difficult to construct a 2-hop spanner for $F(h)$ with $O(n \log n/\log \log n)$ edges, so the lower bound is tight for this construction. It remains an open problem to close the gap between this lower bound and the $O(n \log n)$ upper bound (Theorem~\ref{thm:sq}) for the size of 2-hop spanners of axis-aligned squares in the plane. 

\section{Outlook}
\label{sec:con}

We have shown that every $n$-vertex UDG admits a 2-hop spanner with $O(n)$ edges; and this bound generalizes to the intersection graphs of translates of any convex body in the plane. The proof crucially relies on new results on the $\alpha$-hull of a planar point set. It remains an open problem whether these results generalize to higher dimensions, and whether unit ball graphs admit 2-hop spanners with $O_d(n)$ edges in $\R^d$ for any $d\geq 3$.

We proved that the intersection graph of $n$ axis-aligned squares in $\R^2$ admits a 2-hop spanner with $O(n\log n)$ edges, and this bound is tight up to a factor of $\log \log n$. 
Very recently, Chan and Huang~\cite{ChanH23} established the same upper bound for disks of arbitrary radii in $\R^2$, using shallow cuttings.
However, it is unclear whether the upper bound generalizes to 
other fat convex bodies in the plane
or to balls of arbitrary dimensions in higher dimensions.
For fat convex bodies and for axis-aligned rectangles in the plane, we obtained 3-hop spanners with $O(n\log n)$ and $O(n\log^2 n)$ edges, respectively. 
The latter bound was subsequently improved to $O(n\log n)$ by Chan and Huang~\cite{ChanH23}.
However, it is unclear whether the logarithmic factors are necessary. Do these intersection graphs admit weighted edge biclique covers of weight $O(n)$? In general, we do not even know whether a linear bound can be established for any constant stretch: Is there a constant $t\in \N$ for which every intersection graph of $n$ disks or rectangles admits $t$-hop spanner with $O(n)$ edges? 

Finally, it would be interesting to see other classes of intersection graphs (e.g., for strings or convex sets in $\R^2$, set systems with bounded VC-dimension or semi-algebraic sets in $\R^d$) for which the general bound of  $O(n^{1+1/\lceil t/2\rceil})$ edges for $t$-hop spanners can be improved.

\paragraph*{Acknowledgements.}
We thank Sujoy Bhore for helpful discussions on geometric intersections graphs. We are grateful to the reviewers of earlier versions of this paper for many insightful comments and suggestions. 

\bibliographystyle{plain}
\bibliography{main}

\end{document}